\documentclass[sigconf, screen, authorversion]{acmart}
\setcopyright{acmcopyright}
\copyrightyear{2022}
\acmYear{2022}
\acmConference[SPAA'22]{SPAA '22: ACM Symposium on Parallelism in Algorithms and Architectures}{July 11--14, 2022}{Philadephia, PA}
\usepackage[ruled,noend,nofillcomment]{algorithm2e}
\SetKwComment{comment}{}{}

\SetCommentSty{commentfont}
\SetKwInput{KwAssume}{Assumes}
\SetKwInput{KwRequire}{Requires}
\DontPrintSemicolon
\SetArgSty{relax}

\usepackage{todonotes}
\usepackage{amsmath}
\usepackage{xspace}
\usepackage{tikz}
\usetikzlibrary{calc,backgrounds,snakes}
\usepackage{xcolor}
\newcommand{\card}[1]{\ensuremath{\left|#1\right|}\xspace}
\newcommand{\range}[2][0]{\ensuremath{[|#1, #2|]}\xspace}
\newcommand{\pluseq}{\mathrel{+{=}}}
\newcommand{\minuseq}{\mathrel{-{=}}}
\newcommand*{\transpose}[1]{#1^{\mkern-1.5mu\mathsf{T}}}
\newcommand{\bigo}[1]{\ensuremath{\mathcal{O}(#1)}\xspace}

\newcommand{\floor}[1]{\ensuremath{\left\lfloor #1 \right\rfloor}\xspace}

\newcommand{\restr}[2]{\ensuremath{#1_{|#2}}\xspace}
\newcommand{\footprint}[1]{\ensuremath{\tau\left(#1\right)}\xspace}
\newcommand{\set}[2]{\ensuremath{\left\{ #1 \,|\, #2 \right\}}\xspace}
\newcommand{\mat}[1]{\ensuremath{\mathbf{#1}}\xspace}
\newcommand{\TB}{\ensuremath{\text{TB}}\xspace}
\newcommand{\tbs}{{TBS}\xspace}
\newcommand{\lbc}{{LBC}\xspace}
\newcommand{\oocsyrk}{{OCS}\xspace}
\newcommand{\ooctrsm}{{OCT}\xspace}
\newcommand{\oocchol}{{OCC}\xspace}
\newcommand{\qty}[2]{\ensuremath{Q_{\text{#1}}\left(#2\right)}\xspace}
\newcommand{\qtyvals}[2]{\ensuremath{\widetilde{Q}_{\text{#1}}\left(#2\right)}\xspace}
\newcommand{\alg}[1]{\textbf{#1}}
\newcommand{\TBS}{\alg{TBS}\xspace}
\newcommand{\LBC}{\alg{LBC}\xspace}
\newcommand{\OOCSYRK}{\alg{OOC\_SYRK}\xspace}
\newcommand{\OOCTRSM}{\alg{OOC\_TRSM}\xspace}
\newcommand{\OOCCHOL}{\alg{OOC\_CHOL}\xspace}

\begin{document}

\title{I/O-Optimal Algorithms for Symmetric Linear Algebra Kernels}

\author{Olivier Beaumont}
%% \authornote{Both authors contributed equally to this research.}
\email{olivier.beaumont@inria.fr}
%%\orcid{1234-5678-9012}
\author{Lionel Eyraud-Dubois}
\email{lionel.eyraud-dubois@inria.fr}
\author{Mathieu Vérité}
\email{mathieu.verite@inria.fr}
%%\authornotemark[1]
\affiliation{%
  \institution{Univ. Bordeaux, CNRS, Bordeaux INP, Inria, LaBRI, UMR 5800}
  \city{Talence}
  \country{France}
  \postcode{F-33400}
}

\author{Julien Langou}
\affiliation{%
  \institution{University of Colorado Denver}
  %%\streetaddress{1 Th{\o}rv{\"a}ld Circle}
  \city{Denver}
  \state{Colorado}
  \country{USA}}
\email{julien.langou@ucdenver.edu}

%%\renewcommand{\shortauthors}{Trovato and Tobin, et al.}

\begin{abstract}
  In this paper, we consider two fundamental symmetric kernels in
  linear algebra: the Cholesky factorization and the symmetric
  rank-$k$ update (SYRK), with the classical three nested loops algorithms
  for these kernels. In addition, we consider a machine model with a fast memory of size $S$ and an unbounded slow memory. In this model, all computations must be performed on operands in fast memory, and the goal is to minimize the amount of communication between slow and fast memories.  As the set of computations is fixed by the choice of the algorithm, only the ordering of the computations (the schedule) directly influences the volume of communications.
  
  We prove lower bounds of $\frac{1}{3\sqrt{2}}\frac{N^3}{\sqrt{S}}$
  for the communication volume of the Cholesky factorization of an
  $N\times N$ symmetric positive definite matrix, and of
  $\frac{1}{\sqrt{2}}\frac{N^2M}{\sqrt{S}}$ for the SYRK computation
  of $\mat{A}\cdot\transpose{\mat{A}}$, where $\mathbf{A}$ is an $N\times M$
  matrix. Both bounds improve the best known lower bounds 
  from the literature by a factor $\sqrt{2}$.

  In addition, we present two out-of-core, sequential algorithms with
  matching communication volume: \TBS for SYRK, with a volume of
  $\frac{1}{\sqrt{2}}\frac{N^2M}{\sqrt{S}} + \bigo{NM\log N}$, and \LBC for
  Cholesky, with a volume of $\frac{1}{3\sqrt{2}}\frac{N^3}{\sqrt{S}}
  + \bigo{N^{5/2}}$. Both algorithms improve over the best known
  algorithms from the literature by a factor $\sqrt{2}$, and prove
  that the leading terms in our lower bounds cannot be improved further. This work shows
  that the operational intensity of symmetric kernels like SYRK or
  Cholesky is intrinsically higher (by a factor $\sqrt{2}$) than that
  of corresponding non-symmetric kernels (GEMM and LU factorization).
\end{abstract}

% TO footnote:
%  DONE 3S is the best possible bound
%  DONE g \leq q is expected to be lower (cf Example 1.5 of Opera de Cribro)
%  DONE recursion dans LBC: ça change pas le terme dominant
%  conclusion: améliorer les termes en N^5/2 dans l'algo, inclure des termes d'ordre inférieur dans la borne?




%% \begin{teaserfigure}
%%   TODO
%% \end{teaserfigure}

\maketitle

\section{Introduction}
\label{sec.intro}

Dense matrix factorizations play a significant role in scientific computing. In
particular, symmetric positive definite matrices appear in many applications; the
Cholesky factorization is a dedicated factorization algorithm for such matrices. 
With the increase of both the scale of the platforms and the problems to solve, 
minimizing communications for scientific computing, and in particular for these factorization
kernels, is crucial to reach the peak performance from modern
hardware. In addition to its effect on execution time, the volume of data movement 
has a major impact on energy consumption during a computation, so that 
reducing the volume of data movement will result in a reduction of the 
energy required by a computation.

To evaluate data movements, we consider in this paper an elementary model from the literature, with a computation unit associated to a fast memory of size $S$ where the operands of computations must be located, and a slow memory. In this context, the objective is to perform a set of computations, given by the algorithm, while minimizing the data movements between the fast memory and the slow memory. 
Different orderings (schedules) for a given set of computations (algorithm) 
may induce different volume of
communication between slow and fast memories. We are interested both in establishing the minimum
volume of communication needed for a given computation and in finding
optimal schedules.

The \emph{operational intensity} (OI), defined as the ratio of the
number of arithmetic operations to the volume of data movement to/from
memory, is a critical metric for comparing the efficiency of various
algorithms and their schedules. Since the number of operations for our
computation is known and fixed, studying the operational intensity is
similar to studying the volume of data movement. We want to increase
OI, we do so by reducing the volume of communication.

In this paper, we study the OI of the Cholesky
factorization and of the symmetric rank-$k$ update (SYRK): we derive
lower bounds of the volume of data movement that any schedule has to
perform for each of these algorithms.  Such bounds are valid both for
parallel algorithms and out-of-core sequential algorithms, and we
provide out-of-core sequential algorithms whose communication
complexities asymptotically match our lower bounds.

It is known~\cite{ballard2010communication} that performing a Cholesky
factorization of an $N\times N$ matrix with a memory of size $S$
requires $\Omega\left(\frac{N^3}{\sqrt{S}}\right)$ data
transfers. Regarding the optimal constant,
on the one hand, recent work on automated lower bound
derivation~\cite{olivry2020automated} show that the constant is at
least $\frac{1}{6}$; and, on the other hand, an algorithm by
B\'ereux~\cite{doi:10.1137/06067256X}
proves that the constant is at most $\frac{1}{3}$. 
This represents a factor of 2 between the largest known lower bound 
and the smallest known upper bound.
In this paper, we close this gap by increasing the lower bound by a
factor $\sqrt{2}$ and decreasing the upper bound by the same factor
therefore reaching optimality.

This also shows that the algorithm of B\'ereux~\cite{doi:10.1137/06067256X} is
not optimal which is a surprising result: a rule of thumb, until this present paper,
has been that the minimum volume of data transfer, in dense linear operations,
is ``number of operations divided by $\sqrt{S}$``. For the Cholesky
factorization, this would give the constant to be $\frac{1}{3}$ and make
B\'ereux's algorithm~\cite{doi:10.1137/06067256X} optimal.

Similar results exist for the SYRK kernel, which computes the symmetric matrix
$\mat{C} = \mat{A}\cdot\transpose{\mat A}$ where \mat{A} is a $N \times M$
matrix. With a memory of size $S$, $\Theta\left(\frac{N^2M}{\sqrt{S}}\right)$
transfers are required. The \OOCSYRK algorithm by
B\'ereux~\cite{doi:10.1137/06067256X} achieves a constant of $1$, and a recent
work~\cite{olivry2020automated} provides a lower bound with a constant of
$\frac{1}{2}$. Again, it was believed that the optimal value was
$1$. Again, in this paper, we prove that the best known lower  
bound~\cite{olivry2020automated} can be increased by a factor of $\sqrt{2}$ and
the I/O volume of the best known algorithm~\cite{doi:10.1137/06067256X} 
decreased by the same factor, hence reaching optimality.

Previous works on the SYRK kernel take the symmetry of the operation
$\mat{C} = \mat{A}\cdot\transpose{\mat A}$ into account by computing
the lower half of the matrix $\mathbf{C}$. However the fact that the
matrix \mat{A} appears twice on the right-hand side has not been
exploited: if $A_{j,i}$ is required while $A_{i,j}$ is in the fast
memory, the schedule loads $A_{j,i}$ anyway. Indeed it is akin to
computing the lower half of $\mat{C} = \mat{A}\cdot\mat B$, where
\mat{A} is a $N \times M$ matrix and \mat{B} is a $M \times N$
matrix. Similarly, a lower bound for Cholesky is derived
in~\cite{kwasniewski2021parallel2} under the constraint that it is
forbidden to use $A_{i,j}$ when $A_{j,i}$ is available. The lower bound obtained with
such a constraint is potentially too large: there may exist algorithms
which perform fewer data transfers by making a better use of symmetry.

In this paper, we precisely show how to exploit the symmetry of input to
evaluate the volume of I/O actually required to perform SYRK and Cholesky and
derive lower bounds (thus potentially lower than bounds that do not take it
into account). We also present algorithms which make explicit use of the
symmetry to reduce data movement.

%Previous work on the SYRK kernel takes the symmetry of the operation $\mat{C} =
%\mat{A}\cdot\transpose{\mat A}$ into account by computing the lower half
%of the matrix $C$. However the fact that the matrix \mat{A} appears twice on
%the right-hand side has not been exploited. Indeed, previous work on the SYRK kernels are akin
%to computing the lower half of $\mat{C} = \mat{A}\cdot\mat B$, where \mat{A} is
%a $N \times M$ matrix and \mat{B} is a $M \times N$ matrix. In this paper, we
%show how to exploit the fact that $\mat{B} = \transpose{\mat A}$ to reduce the
%volume of communication by $\sqrt{2}$. 

After a more detailed presentation of the related works in
Section~\ref{sec.related} and a presentation of the methodology in
Section~\ref{sec:methodology}, we thus present our 4 contributions:
\begin{itemize}
\item an improvement by a factor $\sqrt{2}$ of the best known lower
  bound for the communication requirements of the SYRK kernel (from
  $\frac{1}{2}$ to $\frac{1}{\sqrt{2}}$,
  Section~\ref{subsec:bounds.syrk});
\item an application of this result to the Cholesky factorization,
  improving over the best known lower bound by a factor
  $\sqrt{2}$ (from $\frac{1}{6}$ to $\frac{1}{3\sqrt{2}}$,
  Section~\ref{subsec:chol.lower.bound});
\item a \TBS algorithm for the SYRK kernel which makes use of the
  symmetry of the computations to reduce the communication volume by a
  factor $\sqrt{2}$ over previous approaches (from $1$ to
  $\frac{1}{\sqrt{2}}$, Section~\ref{subsec:tbs});
\item a \LBC algorithm for the Cholesky factorization, which uses \TBS
  and a large-block, right-looking approach to obtain a
  communication-optimal schedule (Section~\ref{subsec:lbc}),
  providing a $\sqrt{2}$ improvement over the best known algorithm
  (from $\frac{1}{3}$ to $\frac{1}{3\sqrt{2}}$).
\end{itemize}

Our results provide a proof that the maximal operational intensity for
the multiplication operations in SYRK and Cholesky is
$\sqrt{\frac{S}{2}}$, or equivalently $\sqrt{2S}$ when also counting
the addition operations. This is to be compared to the equivalent
results for matrix-matrix multiplication (GEMM) and LU factorization,
for which the maximal operational intensity is ${\sqrt{S}}$ (see
Table~1 in~\cite{olivry2020automated}). Our work shows that symmetric
operations have intrinsically higher operational intensity, and our
algorithms provide insight about how to take advantage of that. At
this stage, we do not claim that our algorithms can be used in a
practical setting, but their existence shows that the lower bounds
that we obtain are the best possible.

%% Our main result applies to the SYRK kernel, for which we obtain an
%% asymptotically communication-optimal algorithm. Since most of the computing
%% operations of the Cholesky factorization is performed in SYRK updates, we use
%% this result to obtain a communication-optimal Cholesky algorithm.

\section{Related Works}

\label{sec.related}

%% We first present related works that prove lower bounds on the volume
%% of communication required to perform linear algebra operations and
%% factorizations, given a memory size $S$. Then, we describe existing
%% communication-avoiding algorithms to perform these operations, and
%% their obtained performance in terms of volume of communication.

Research work regarding the estimation of data transfers required to
perform linear algebra kernels heavily rely on two simplified machine
models:

\begin{enumerate}
  \item Two-levels memory model: the machine features one fast and
    limited memory of size $S$, one ``slow'' and unlimited memory. 
    Input required for any computation must reside in fast memory to 
    be performed, while the data initially resides in slow memory.

  \item Parallel model: $P$ nodes, each with a memory of size 
    $S$, can communicate through a network.
\end{enumerate}

Both models are highly related: the two-level model can be used to
study the volume of communication of a single node in a parallel
machine, since the set of all other nodes can be viewed as a single
"slow" memory with which data transfers occur. Thus, most lower bounds
are actually obtained in the two-level model, and then transferred to the
parallel model. For the purpose of comparison, only bandwidth
communication cost estimations are of interest to us: we do not
consider latency issues or bound the number of messages, we focus on
the number of data elements transferred.

\subsection{Two-level sequential model}

The paper from Hong and Kung~\cite{Jia-Wei:1981:ICR:800076.802486} can be considered as
the founding piece of subsequent development on the topic. In this
work, the authors consider a two-level memory machine. A set of rules
referred as the \emph{pebbling game} models the required data
transfers between the two types of memory. Based on the analysis of
the \emph{computational DAG}, the authors derive asymptotic lower
bounds for the number of data transfers required between the two
levels of memory. For classical $N \times N$ matrix multiplication
(\textit{i.e.} requiring $\bigo{N^{3}}$ operation), Hong and Kung
prove that $\Omega\left(\frac{N^{3}}{\sqrt{S}}\right)$ data transfers
are required.

Irony~\emph{et al.} improve this result in \cite{irony2004communication}
through the expression of a ''memory-communication tradeoff'' in the
context of two-levels memory model. For $N \times N$ matrix multiplication,
the total number of data transfers between slow and fast memory actually is
$\Theta(\frac{N^{3}}{\sqrt{S}})$. This however assumes a limited fast memory
size: $S \leqslant \frac{N^{2}}{\sqrt[3]{32}} $.

In \cite{ballard2010communication}, Ballard \emph{et al.} extend
Hong and Kung's result to Cholesky factorization using a reduction
technique: by observing that a $\frac{N}{3} \times \frac{N}{3}$ matrix
multiplication can be carried out through a $N \times N$ Cholesky
factorization, they prove that the communication costs of the latter
method are only a constant factor of the former. Therefore, the
asymptotic bounds established in \cite{Jia-Wei:1981:ICR:800076.802486} hold: performing
a $N \times N$ Cholesky factorization requires at least
$\Omega(\frac{N^{3}}{\sqrt{S}})$ data transfers. This result has been
later generalized to a broader variety of
kernels~\cite{DBLP:journals/siammax/BallardDHS11,ballard_carson_demmel_hoemmen_knight_schwartz_2014}

This line of work enables to establish very general bounds, for a
broad range of kernels, including sparse computations, and provides
algorithms with matching communication complexity. However, these
algorithms are only asymptotically optimal: they achieve a
communication volume $\bigo{\frac{N^{3}}{\sqrt{S}}}$ for Cholesky for
example, but the respective hidden constant factors of the lower
bounds and the algorithms can be significantly different.

Very recently,
automatic \emph{cDAG} analysis techniques have led to refinement of
lower bounds for several kernels at once, meaning that the constant
factor of the dominant term is explicitly provided. In particular,
Olivry~\emph{et al.}~\cite{olivry2020automated,olivry:hal-03200539}
derive lower bounds on data transfers (in the context of the out-of-core model)
for any kernel expressed as an affine program. Among other results,
they establish that Cholesky factorization requires at least
$\frac{N^{3}}{6\sqrt{S}}+\bigo{N^{2}}$ I/O operations, and that SYRK
requires at least $\frac{1}{2}\frac{N^2M}{\sqrt{S}}+\bigo{NM}$ I/O
operations. This work also presents a tool which computes an efficient
tiling scheme according to the available memory size. The tool is however
limited to rectangular tilings.

Independently, using explicit enumeration of data reuse, Kwasniewski~\emph{et
 al.}~\cite{kwasniewski2021parallel2} obtain a corresponding lower
bound for LU factorization: their proof is in a parallel context, but
their arguments show that the
minimum number of data transfers is lower bounded by
$\frac{2}{3}\frac{N^{3}}{\sqrt{S}}$. They also propose a
generalization to Cholesky factorization and obtain an improved
$\frac{1}{3}\frac{N^{3}}{\sqrt{S}}$ lower bound, which however makes
the implicit assumption that there is no data reuse related to the
symmetry of the matrix as discussed in Section~\ref{sec.intro}.

%For the SYRK kernel for example, if we do not exploit symmetry, we
%would compute only half of $C$, but we would refuse to use $A_{i,j}$
%when $A_{j,i}$ is needed. So if we need $A_{j,i}$ while $A_{i,j}$ is
%in cache, the schedule loads $A_{j,i}$ anyway.  Similarily, and more
%problematic, the lower bound is derived on the fact that it is
%forbidden to use $A_{i,j}$ when $A_{j,i}$ is available. In some sense
%the algorithm would be akin to computing the lower part of $\mat C
%\pluseq \mat A\cdot\mat B$. The lower bound obtained with such a
%constraint is potentially too large: there may exist algorithms which
%perform fewer data transfers by using this symmetry. Our work shows
%that such algorithms actually exist.

In 2009, B{\'e}reux~\cite{doi:10.1137/06067256X} proposes a sequential
out-of-core Cholesky algorithm with ``narrow blocks'' that performs
$\frac{N^{3}}{3\sqrt{S}}+ \bigo{N^{2}}$ I/O operations, without making
use of the symmetry of the matrix. This matches the lower bound from
Kwasniewski~\emph{et al.}, showing that this is the best possible
bound in the context of this implicit assumption.  The same paper also
mentions an out-of-core SYRK algorithm, based on similar ideas, which
performs $\frac{MN^2}{\sqrt{S}} + \bigo{MN}$ I/O operations.

%Our lower bound derivation allows for data reuse related to the
%symmetry of the matrix (so is potentially lower than a lower bound
%that do not). Our algorithms make explicit use of the symmetry to
%reduce data movement, and improve over previous results for SYRK and
%Cholesky by a factor $\sqrt{2}$.

%% In Section~\ref{sec:bounds}, we prove that when taking into
%% account the possibility of data reuse thanks to the symmetry of the
%% matrix, the minimum number of I/O operations is lower bounded by
%% $\frac{1}{3\sqrt{2}}\frac{N^{3}}{\sqrt{S}}$ for Cholesky and by
%% $\frac{1}{\sqrt{2}}\frac{N^2M}{\sqrt{S}}$ for SYRK. This represents a
%% $\sqrt{2}$ improvement over the corresponding bounds of Olivry
%% \emph{et al.}.

%% In Section~\ref{sec:algs}, we propose two algorithms: \TBS (Triangular
%% Block SYRK) for the SYRK operation, which performs
%% $\frac{MN^2}{\sqrt{2}\sqrt{S}} + \bigo{MN}$ I/O operations; and \LBC
%% (Large Blocks Cholesky), which performs
%% $\frac{1}{3\sqrt{2}}\frac{N^{3}}{\sqrt{S}}+ \bigo{N^{5/2}}$ I/O
%% operations. This improves both already existing algorithms by a factor
%% $\sqrt{2}$, and matches the constant factors of our improved lower
%% bounds.

%% \todo[inline]{JL: the $\bigo{N^{5/2}}$ is a little big. Is there something to temper it? Like a $S$ something?
%% I'll have a look. Always an issue with this $M$, $N$ and $S$. What goes to infinity?} 

\subsection{Parallel model}

In~\cite{irony2004communication} Irony~\emph{et al.} apply their
''memory-communication tradeoff'' for matrix multiplication in the context
of parallel execution. It states that using $P$ nodes to perform the
multiplication of $M \times N$ and $N \times R$ matrices, at least one node must
send or receive at least $\frac{M N R}{2\sqrt{2}P\sqrt{S}}-S$ data.
The authors also present $2D$ and $3D$ task distributions for matrix
multiplication as parallel implementation in the two limit cases for
memory size: $S = \mathcal{O}\left(\frac{N^{2}}{P}\right)$ for the
first case, $S =
\mathcal{O}\left(\frac{N^{2}}{P^{\frac{2}{3}}}\right)$ for the
second. They prove that those algorithms are asymptotically optimal
regarding communications since they match the lower bound derived from
the ''memory-communication tradeoff''. More recent work by Solomonik
\emph{et al.}~\cite{10.1145/2897188} presents an original way of
modeling dependencies of any \emph{cDAG} as a lattice-hypergraph which
enables the authors to extend the ''memory-communication tradeoff'' to
take into account synchronization and express bounds about the
communication on the critical path.

In 2011, Solomonik and Demmel~\cite{solomonik2011lu25D} introduced the
$2.5D$ algorithms for matrix multiplication and LU factorization,
bringing a continuum between $2D$ and $3D$ algorithms.%%  where the trade-off
%% between memory footprint and communication is controlled by a parameter.
Experimental results show the superiority of $2.5D$ algorithms over
conventional $2D$ algorithms.

Regarding Cholesky, Ballard \emph{et al.}~\cite{ballard2010communication}
reviewed existing parallel distributed algorithms and, based on their lower
bound on communication, proved that LAPACK and other block recursive
implementations are asymptotically optimal for a carefully selected block
size. The work on lower bounds by Kwasniewski\emph{et
  al.}~\cite{kwasniewski2021parallel2} leads to the design of parallel
distributed $2.5D$ LU (COnfLUX) and Cholesky (COnfCHOX)
algorithms. These algorithms perform a volume of communication per
node of $\frac{N^{3}}{P\sqrt{S}}+ \bigo{N^{2}}$.

%% Current state of the art:
%% \begin{itemize}
%% \item Langou and Rastello have automatic derivations of lower bounds,
%%   which show that SYRK requires
%%   $\frac{1}{2}\cdot\frac{MN^2}{\sqrt{S}}$ communications, and that
%%   Cholesky requires  $\frac{1}{6}\cdot\frac{N^3}{\sqrt{S}}$.
%% \item Kwasniewski et al. study the LU factorization~\cite{kwasniewski2021parallel2} (non-symmetric
%%   variant of Cholesky), and apply their methodology to Cholesky to
%%   obtain a proof that Cholesky requires
%%   $\frac{1}{3}\cdot\frac{N^3}{\sqrt{S}}$. However, this proof makes
%%   the implicit assumption that algorithms do not benefit from the
%%   possibility of data re-use which comes from the symmetry of the
%%   matrix.
%% \item Bereux proposed a matching algorithm with
%%   $\frac{1}{3}\cdot\frac{N^3}{\sqrt{S}} + \bigo{N^2}$ communications,
%%   which indeed does not use the symmetry of the matrix. So with that
%%   assumption the bound is tight.
%% \item We can apply the same idea to the SYRK case: load the result
%%   matrix by tiles of $\sqrt{S}\times \sqrt{S}$, each tile once, and
%%   for each tile, load $2M$ vectors of length $\sqrt{S}$. This yields an
%%   algorithm with $\frac{MN^2}{\sqrt{S}}$ communications.
%% \end{itemize}

%% Our work thus improves all these results by a factor $\sqrt{2}$.

\section{Assumptions and Methodology}
\label{sec:methodology}
We consider a computational platform with a slow memory of unbounded
size, and a fast memory of bounded size $S$. We fix a given
computation described with a computational directed acyclic graph cDAG
$G=(V, E)$, where each vertex in $V$ represents a computation
operation and each edge in $E$ represents a data dependency between
operations. An operation can only be performed if the corresponding
input data is in fast memory.  We assume that the algorithms
explicitly control which data is loaded and removed from the fast
memory. The operations in $V$ can be performed in different orders,
and we are interested in finding orderings which induce the minimum
amount of transfers between slow and fast memory, also called
\emph{I/O operations}.

\subsection{Lower bound methodology}

The lower bounds of this paper are based on a careful application of
Lemma~1 in~\cite{kwasniewski2021parallel2}, which states:
\begin{lemma}
  Fix a constant $X>S$ and assume that any subcomputation $H$ of a
  cDAG $G=(V, E)$ which reads at most $X$ elements and writes at most
  $X$ elements performs a number of operations $\card{H}$ bounded by
  $\card{H}\leqslant H_{\max}$.

  Consider any execution of $G$ with memory $S$. Its operational
  intensity $\rho$ is bounded by $\rho \leqslant \frac{H_{\max}}{X-S}$, and
  its number of I/O operations $Q$ is bounded by
  $$ Q \geqslant \frac{\card{V}}{\rho} \geqslant
  \frac{\card{V}(X-S)}{H_{\max}}.$$
  \label{lemma:kwasniewski}
\end{lemma}

In~\cite{kwasniewski2021parallel2}, the number of elements read and
written by a subset of computations $H$ are expressed in terms of dominator
sets and minimum sets in the graph $G$. In our case however, the graph
is quite regular, so we do not need to introduce these notions.

\begin{algorithm}[t]
  % \KwIn{$\mat{A}$ of size $N\times M$, $\mat{C}$ symmetric of size $N\times N$, only lower triangular part of $C$ is referenced}
  % \KwOut{$\mat{C} \pluseq \mat{A}\cdot \transpose{\mat{A}}$, only lower triangular part of $C$ is computed}
  \KwIn{$\mat{A}$ of size $N\times M$, $\mat{C}$ symmetric of size $N\times N$}
  \KwOut{$\mat{C} \pluseq \mat{A}\cdot \transpose{\mat{A}}$}
  \For{$i = 1$ to $N$}{
    \For{$j = 1$ to $i$}{
      \For{$k = 1$ to $M$}{
        $\mat{C}_{i,j} \pluseq \mat{A}_{i,k}\cdot \mat{A}_{j, k}$\;
      }
    }
  }
  \caption{Pseudo-code of SYRK, where only the lower triangular part
    of \mat{C} is referenced and computed.}
  \label{alg:syrk}
\end{algorithm}
\begin{algorithm}[t]
  % \KwIn{$\mat{A}$ symmetric positive definite of size $N\times N$, only lower triangular part of $A$ is referenced}
  % \KwOut{Replace lower triangular part of $\mat{A}$ with lower triangular matrix $L$ such that $\mat{A} = {\mat{L}}\cdot \transpose{\mat{L}}$}
  \KwIn{$\mat{A}$ symmetric positive definite of size $N\times N$}
  \KwOut{Replace $\mat{A}$ with $\mat{L}$ such that $\mat{A} = {\mat{L}}\cdot \transpose{\mat{L}}$}
  \For{$k = 1$ to $N$}{
    $\mat{A}_{k, k} = \sqrt{\mat{A}_{k, k}}$\;
    \For{$i = k+1$ to $N$}{
      $\mat{A}_{i, k} = \mat{A}_{i, k} / \mat{A}_{k, k}$\;
      \For{$j = k+1$ to $i$}{
        $\mat{A}_{i, j} \minuseq \mat{A}_{i, k}\cdot \mat{A}_{j, k}$\comment*{update operations}
      }
    }
  }
  \caption{Pseudo-code of Cholesky, where only the lower triangular
    parts of \mat{A} and \mat{L} are referenced.}
  \label{alg:cholesky}
\end{algorithm}

We consider the SYRK and Cholesky kernels, as described in
Algorithms~\ref{alg:syrk} and~\ref{alg:cholesky}. In the following,
$N$ and $M$ always denote the sizes of the matrices used in the
kernels. In the Cholesky kernel, for the lower bound target we will
focus on the \emph{update operations} only.  We can thus describe each
operation by a triplet of positive integers $(i, j, k)$, and
for both cases we will further ignore the diagonal operations where $i
= j$. The sets of operations are denoted $\mathcal{S}$ for the SYRK
kernel and $\mathcal{C}$ for Cholesky, and are given by:
\begin{align*}
  \mathcal{S} &= \set{(i, j, k) \in \range[1]{N}^2\times \range[1]{M}}{i > j}\\
  \mathcal{C} &= \set{(i, j, k) \in \range[1]{N}^3}{i > j > k},
\end{align*}
where \range[a]{b} denotes the set of integers between $a$ and $b$
(inclusive).

We can see that for each statement of these algorithms, the set of
written variables is included in the set of read variables, so we only
focus on the input data of each operation. In the rest of the paper
$H$ is used to denote a set of operations, subset of $\mathcal{S}$ or
$\mathcal{C}$.
\begin{definition}
  Given a set $H$ of operations, \restr{H}{k} is the restriction of $H$
  to iteration $k$:
  $$\restr{H}{k} = \set{(i, j) \in \mathbb{N}^{2}}{(i, j, k) \in H}.$$
\end{definition}
%% \begin{definition}
%%   $p_{3}(H)$ is the \emph{projection} of $H$ along the 3rd dimension:
%%   $$p_3(H) = \{ (i,j) \in \mathbb{N}^2 \; | \; \exists k, (i, j, k) \in
%%   H\} = \bigcup_{k} \restr{H}{k}.$$
%% \end{definition}
\begin{definition}
  Given a subset $U$ of $\mathbb{N}^2$, $\footprint{U}$ is the
  \emph{symmetric footprint} of $U$:
  $$\footprint{U} = \set{i \in \mathbb{N}}{\exists j, (i, j) \in U
  \text{ or } (j, i) \in U}.$$
  If $i > j$ for all $(i, j)\in U$, then $\card{U} \leqslant
  \frac{\card{\footprint{U}}(\card{\footprint{U}}-1)}{2}$. In particular, this holds for
  any $\restr{H}{k}$.
\end{definition}

With these definitions, we can express the number of data accessed by
a set $H$: using the SYRK kernel as an example, $\bigcup_k
\restr{H}{k}$ is the set of elements $C_{i,j}$ accessed by $H$, and
for any $k$, $\footprint{\restr{H}{k}}$ is the set of $A_{i, k}$ elements
accessed by $H$.
%% \todo[inline]{LED: Do not know how to write that $A_{i,k}$ elements are different for different values of $k$, which is why we sum the cardinals and not consider the union.}

\begin{proposition}
  For any set $H$ of operations, the number of data accessed by $H$ is
  $$D(H) = \card{\cup_k \restr{H}{k}} + \sum_k \card{\footprint{\restr{H}{k}}}.$$
\end{proposition}

\subsection{Triangle blocks}

Many results in this paper are obtained by considering \emph{triangle
  blocks}, which are generalizations of the diagonal tiles in a tile
decomposition of a symmetric matrix. In particular, the SYRK lower
bound shows that accessing the result matrix by triangle blocks is the
most efficient, and the TBS algorithm describes how to partition the
result matrix in disjoint triangle blocks. Figure~\ref{fig:tbs.zones}
(page~\pageref{fig:tbs.zones}) depicts examples of triangle blocks.

\begin{definition}[Triangle block]
  \label{def:triangle.block}
  Given a set $R$ of integer indices, the \emph{triangle block}
  $\TB(R)$ is the set of all subdiagonal pairs of elements of $R$:
  \begin{align*}
  %% TB(R) =
  %% \binom{R}{2} &= \left\{\{r, r'\} \, | \, r , r' \in R \text{ and } r
  %% \neq r'\right\}\\
    \TB(R) &= \set{(r, r')}{r, r' \in R \text{ and } r > r'}
  \end{align*}
\end{definition}
It is clear that $\card{\TB(R)} = \frac{\card{R}(\card{R}-1)}{2}$. We
say that $\TB(R)$ has side length $\card{R}$.

For any $m\in \mathbb{N}$, we define $\sigma(m)$ as the smallest
possible side length of a triangle block with at least $m$ elements.
$\sigma(m)$ is thus the smallest element of $\mathbb{N}$ such that $m
\leqslant \frac{\sigma(m)(\sigma(m)-1)}{2}$. By solving the quadratic
equation, we get:

\begin{lemma}
  For $m\in \mathbb{N}^*$, $\sigma(m)= \lceil
  \sqrt{\frac{1}{4}+2m}+\frac{1}{2} \rceil$, and $\sigma(0) = 0$.
\end{lemma}

For any $m \in \mathbb{N}$, we define $T(m)$ as any size-$m$ subset of
$\TB(\range[1]{\sigma(m)})$. We use $T(m)$ as a canonical way of
performing $m$ computations in an iteration, while minimizing the
number of data accesses. Indeed, by definition $\card{T(m)} = m$, and
it is easy to see that $\card{\footprint{T(m})} = \sigma(m)$.

\section{Lower Bounds}
\label{sec:bounds}

\subsection{Symmetric Multiplication (SYRK)}
\label{subsec:bounds.syrk}

As mentioned above, in order to obtain a lower bound on the data
movements required for the SYRK computation, we first provide an upper
bound on the largest subcomputation $H$ than can be performed while
accessing at most $X$ data elements. We are thus looking for (a bound on) the
optimal value of the following optimization problem:
\begin{align*}
  \mathcal{P}(X)\text{:} \quad& \max \; \card{H}\\
   \textrm{s.t.} \quad& 
   D(H) = \card{\cup_k \restr{H}{k}}+\sum_{k = 1}^{M} \card{\footprint{\restr{H}{k}}} \leqslant X \\
   & H \subseteq \mathcal{S}%%\{(i, j, k) \in \range[1]{N}^2\times\range[1]{M} \,|\, i > j \}
\end{align*}

The main result of this section can be stated as:

\begin{theorem}
  The optimal value of $\mathcal{P}(X)$ is at most
  $\frac{\sqrt{2}}{3\sqrt{3}}X^{\frac{3}{2}}$.
  \label{thm:syrk.main.bound}
\end{theorem}

To prove this theorem, we first show that $\mathcal{P}(X)$ admits
triangle-shaped optimal solutions, which we call \emph{balanced
  solutions}. 
We then compute an upper bound on the size of such a
balanced solution.

\subsubsection{Balanced Solutions}

\begin{definition}
  For given $x$ and $m$, we define the \emph{balanced solution} $B = B(x, m)$ by:
  \begin{displaymath}
    \begin{cases}
       \restr{B}{k}  = T(m) & \text{for all }k \in \range{K-1},\\
       \restr{B}{K} = T(m')& \text{for } k = K,\\
       \restr{B}k = \emptyset & \text{for all }k > K, \\
    \end{cases}
  \end{displaymath}
  where $K = \lfloor \frac{x}{m}\rfloor$ and $m' = x - Km < m$.
\end{definition}

It is clear that $\card{B(x, m)} = x$ (since $K\cdot m + (x-Km) = x$)
and $\card{\cup_k \restr{B(x, m)}{k}} = m$. The next lemma shows that
any solution $H$ can be turned into a balanced solution with lower
cost.

\begin{lemma}
  If $H$ is a solution of $\mathcal{P}(X)$, let the corresponding
  balanced solution be $B = B\left(\card{H}, \max_k
  \card{\restr{H}{k}}\right)$. Then $D(B)\leqslant D(H)$.
\end{lemma}
\begin{proof}
  Given a solution $H$, let us define $m_k = \card{\restr{H}{k}}$ and
  denote $m = \max_k m_k$. As mentioned above, we have $\card{B} =
  \card{H}$ and $\card{\cup_k \restr{B}{k}} = m = \max_k m_k \leqslant
  \card{\cup_k \restr{H}{k}}$. Furthermore, since $\sum_k m_k =
  \card{H} = K\cdot m + m'$ and since the $\sigma(\cdot)$ function is
  concave, we have:

  \begin{align*}
    \sum_{k}\card{\footprint{\restr{B}{k}}}& = K\sigma(m)+\sigma(m') \\
    & \leqslant \sum_{k}\sigma(m_{k}) \\
    & = \sum_{k}\card{\footprint{\restr{H}{k}}}
  \end{align*}

  This shows that $D(B) \leqslant D(H)$.
\end{proof}

In particular, if $H$ is an optimal solution, we obtain the following
corollary.
\begin{corollary}
  There exist $x$ and $m$ such that $B(x, m)$ is an optimal solution
  to $\mathcal{P}(X)$.
  \label{lem:balanced.is.optimal}
\end{corollary}

\subsubsection{Optimal Balanced Solution}

A balanced solution $B$ can be described with three integer values $I$, $J$
in \range[1]{N} with $J\leqslant I$, and $K\in \range[1]{M}$ such that

\begin{displaymath}
  \left\lbrace \begin{array}{l}
    \forall k \in \range{K-1}, \; B_{k} = T(I) \\
    B_{K} = T(J)
  \end{array} \right.
\end{displaymath}

Such a solution satisfies $\card{B} = K\frac{I(I-1)}{2} +
\frac{J(J-1)}{2}$ and $D(B) = \frac{I(I-1)}{2}+KI+J$. By relaxing
integrity constraints and upper bounds on $I, J, K$, we get that the
optimal size of a balanced solution is at most the optimal value of
the following problem $\mathcal{P}'(X)$:

\begin{align*}
  \mathcal{P}'(X)\text{:}\quad &\max \left( K\frac{I(I-1)}{2} + \frac{J(J-1)}{2}\right)\\
  & \textrm{s.t.} \quad
  \left\lbrace \begin{array}{l}
    \frac{I(I-1)}{2}+KI+J \leqslant X \\
    J \leqslant I \\
  \end{array} \right.
\end{align*}

\begin{lemma}
  For any $(I, J, K)$ solution to $\mathcal{P}'(X)$, define $K' = K +
  \frac{J(J-1)}{I(I-1)}$. Then $(I, 0, K')$ is a solution to
  $\mathcal{P}'(X)$ with the same value.
  \label{lem:j.equals.zero}
\end{lemma}
\begin{proof}
  The solution $(I, 0, K')$ is feasible:
  \begin{align*}
    \frac{I(I-1)}{2}+K'I &= \frac{I(I-1)}{2}+KI+J\frac{J-1}{I-1}\\
    &\leqslant \frac{I(I-1)}{2}+KI+J &\text{since $J \leqslant I$}\\
    & \leqslant X & \text{since $(I, J, K)$ is feasible}
  \end{align*}

  Furthermore, its objective value is $K'\frac{I(I-1)}{2} =
  K\frac{I(I-1)}{2}+\frac{J(J-1)}{2}$, which is the objective value of
  $(I, J, K)$.
\end{proof}

This lemma shows that the optimum value of $\mathcal{P}'$ is equal to
the optimum value of the simpler $\mathcal{P}''$ problem below:

\begin{align*}
 \mathcal{P}''(X)\text{:}\quad & \max \left( K\frac{I(I-1)}{2} \right)\\
 \textrm{s.t.} \quad &\frac{I(I-1)}{2}+KI \leqslant X \\
\end{align*}

This problem is now simple enough and we can provide a direct bound on
its optimum value.

\begin{lemma}
  The optimum value of $\mathcal{P}''(X)$ is at most
  $\frac{\sqrt{2}}{3\sqrt{3}}X^{\frac{3}{2}}$.
  \label{lem:final.optim}
\end{lemma}
\begin{proof}

  Reformulated as a minimization problem, $\mathcal{P}''(X)$ becomes:
  \begin{align*}
    & \min \left( f(K, I) = -K\frac{I(I-1)}{2} \right) \\
    \textrm{s.t.} \quad &
    g(K, I) = \frac{I(I-1)}{2}+KI -X \leqslant 0 \\
  \end{align*}

  Since the regularity conditions are met over the whole definition
  space of variables $I$ and $K$, we can write Karush-Kuhn-Tucker
  necessary conditions: if $(K, I)$ is a local optimum for
  $\mathcal{P}''(X)$ then
  \begin{align*}
    \exists u \geqslant 0, \quad & \nabla f(K, I)+u \nabla g(K, I) = 0 \\ 
   \Leftrightarrow \exists u \geqslant 0, \quad &
  \left\lbrace \begin{array}{l}
    -K(I-\frac{1}{2})+u(I-\frac{1}{2}+K) = 0 \\
    -\frac{I(I-1)}{2}+uI = 0
  \end{array} \right.
\end{align*}

which implies $u = \frac{I-1}{2}$, and then $KI = (I-1)(I-\frac{1}{2})$.

%% \todo[inline]{MV: how to go from necessary condition of a local
%%   optimum to search of the actual global optimum?}

Let us denote by $(K, I)$ a local minimum of $f$. Then $KI =
(I-1)(I-\frac{1}{2})$. Besides we can select $(K, I)$ such that
$\frac{I(I-1)}{2}+KI-X = 0$. This yields $3I^{2}-4I-(2X-1) = 0$, and
we obtain $I = \frac{2}{3}+\frac{\sqrt{1+6X}}{3}$.

An optimal solution of $\mathcal{P}''(X)$ is thus given by
\begin{displaymath}
  \begin{cases}
    I^{*} = \frac{2}{3}+\frac{\sqrt{1+6X}}{3} \\
    K^{*} = (I^{*}-\frac{1}{2})(1-\frac{1}{I^{*}})
  \end{cases}
\end{displaymath}

and its objective value is
\begin{align*}
    \mathcal{H}''(X) & = K^{*}\frac{I^{*}(I^{*}-1)}{2} \\
    & = \frac{1}{4}(I^{*}-1)^{2}(2I^{*}-1) \\
    & = \frac{1}{108}(\sqrt{1+6X}-1)^{2}(2\sqrt{1+6X}+1) \\
    & \leqslant \frac{(\frac{\sqrt{6X}}{3})^{3}}{2} = \frac{\sqrt{2}}{3\sqrt{3}}X^{\frac{3}{2}}
\end{align*}

To understand why the last inequality holds, one can observe that the
function $X \mapsto \mathcal{H}''(X)-
\frac{\sqrt{2}}{3\sqrt{3}}X^{\frac{3}{2}}$ equals $0$ for $X = 0$.
Besides,

\begin{align*}
  \frac{\partial}{\partial X} \Big[ \mathcal{H}''(X)-
    \frac{\sqrt{2}}{3\sqrt{3}}X^{\frac{3}{2}} \Big]
  & = \frac{1}{6} \big( \sqrt{1+6X}-1 \big) -\sqrt{\frac{X}{6}} \\
  & = \frac{1}{6} \big[ \sqrt{1+6X}-(1+\sqrt{6X}) \big]
\end{align*}

which is obviously negative.
\end{proof}

\subsubsection{Final Result}

\begin{proof}[Proof of Theorem~\ref{thm:syrk.main.bound}]
  The result follows directly from
  Corollary~\ref{lem:balanced.is.optimal},
  Lemma~\ref{lem:j.equals.zero} and Lemma~\ref{lem:final.optim}.
\end{proof}

\begin{corollary}
  The number of data accesses required to perform a SYRK operation
  where $\mat{A}$ has size $N\times M$, with memory $S$, is at least

  $$Q_{\textrm{SYRK}}(N, M, S) \geqslant \frac{1}{\sqrt{2}}
  \frac{N^{2}M}{\sqrt{S}}.$$
  \label{thm:syrk.final.bound}
\end{corollary}
\begin{proof}
Consider the computational DAG of the SYRK operation, which has $\card{\mathcal{S}} =
\frac{N^2M}{2}$ vertices. According to
Theorem~\ref{thm:syrk.main.bound}, for any $X$, any subcomputation $H$
of this DAG which reads at most $X$ elements has size $\card{H}
\leqslant \frac{\sqrt{2}}{3\sqrt{3}}X^{\frac{3}{2}}$.

In particular\footnote{The value $X=3S$ is chosen to obtain the
  strongest possible bound by maximizing the ratio
  $\frac{\card{H}}{X-S}$.}, for $X = 3S$, we get $\card{H} \leqslant
\sqrt{2}\cdot S^{\frac{3}{2}}$. According to
Lemma~\ref{lemma:kwasniewski}, the maximal operational intensity of
SYRK is $\rho = \frac{\card{H}}{3S - S} \leqslant
\sqrt{\frac{S}{2}}$. This yields the following bound on the number of
data accesses for the complete SYRK operation:

  $$Q_{\textrm{SYRK}}(N, M, S) \geqslant \frac{\card{\mathcal{S}}}{\rho} =
\frac{1}{\sqrt{2}} \frac{N^{2}M}{\sqrt{S}}.$$
\end{proof}

\subsection{Cholesky factorization}

\label{subsec:chol.lower.bound}

We now consider the Cholesky factorization, as described by
Algorithm~\ref{alg:cholesky}. As mentioned above, we focus on the
\emph{update operations}, described by the set
$$\mathcal{C} =  \set{(i, j, k) \in \range[1]{N}^3}{i > j > k}.$$

For a given $X$, the largest subset $H$ that accesses at most $X$
elements can be found by solving $\mathcal{P}(X)$, in which the
constraint $H \subseteq \mathcal{S}$ is replaced by $H
\subseteq\mathcal{C}$. We consider a relaxed version, in which the
constraint is instead $H \subseteq \mathcal{C}'$, where
$$\mathcal{C}' =  \set{(i, j, k) \in \range[1]{N}^3}{i > j}.$$

Since $\mathcal{C} \subseteq \mathcal{C}'$, the optimal value of this
relaxed version is not smaller than the optimal value of the original
one. We can now remark that the relaxed version is a special case of
$\mathcal{P}(X)$ where $M = N$, so that we can directly apply
Theorem~\ref{thm:syrk.main.bound}, which leads to the following
corollary:
\begin{corollary}
  The number of data accesses required to perform a Cholesky operation
  on a matrix $\mat{A}$ of size $N\times N$, with memory $S$, is at least
  
  $$Q_{\textrm{Chol}}(N, S) \geqslant \frac{1}{3\sqrt{2}}
  \frac{N^{3}}{\sqrt{S}}.$$
  \label{thm:cholesky.final.bound}
\end{corollary}
\begin{proof}
  The computational DAG of the update operations of the Cholesky
  kernel contains $\card{\mathcal{C}} = \frac{N^3}{6}$ update
  operations. According to Theorem~\ref{thm:syrk.main.bound}, for any
  $X$, any subcomputation $H$ of this DAG which reads at most $X$
  elements has size $\card{H} \leqslant
  \frac{\sqrt{2}}{3\sqrt{3}}X^{\frac{3}{2}}$.

  As previously, we apply Lemma~\ref{lemma:kwasniewski} to
  the case where $X = 3S$, and obtain that the maximal operational
  intensity of the update operations in Cholesky is $\rho =
  \frac{\card{H}}{3S - S} \leqslant \sqrt{\frac{S}{2}}$. Since a Cholesky
  kernel needs to perform all update operations, this yields the
  following bound on the number of data accesses
  $$Q_{\textrm{Chol}}(N, S) \geqslant \frac{\card{\mathcal{C}}}{\rho} =
  \frac{1}{3\sqrt{2}} \frac{N^{3}}{\sqrt{S}}.$$
\end{proof}

\section{Communication-Optimal Algorithms}
\label{sec:algs}

In this section, we propose algorithms which perform the same
operations as Algorithms~\ref{alg:syrk} and~\ref{alg:cholesky}, but
with an ordering that allows to perform fewer I/O operations.  We
start by presenting an algorithm for the SYRK kernel, which we then
use to design an algorithm for the Cholesky kernel.

To simplify the presentation of the algorithms, we
index the matrices in the range \range{N-1} instead of \range[1]{N}.
Our algorithms rely on previously proposed algorithms from
B\'ereux~\cite{doi:10.1137/06067256X}, more specifically the one-tile,
narrow-block variants of \OOCSYRK and \OOCTRSM, and the one-tile,
left-looking variant of Cholesky \OOCCHOL. For conciseness, we denote
them respectively by \oocsyrk, \ooctrsm and \oocchol, with the following number of
I/O operations:
\begin{align*}
  \qty{\oocsyrk}{N, M} &= \frac{N^2M}{\sqrt{S}} + \bigo{NM}\\
  \qty{\ooctrsm}{N, M} &= \frac{N^2M}{\sqrt{S}} + \bigo{NM}\\
  \qty{\oocchol}{N} &= \frac{N^3}{3\sqrt{S}} + \bigo{NM}
\end{align*}

The analysis of communication cost in this section is asymptotic in
the following sense: we assume that $S$ remains constant, and that the
sizes $N$ and $M$ of the matrices grow without bounds.

In the following algorithms, given a matrix \mat{A} and two sets of
indices $X$ and $Y$, we use $\mat{A}[X, Y]$ to denote the submatrix of
$\mat{A}$ indexed with indices in $X \times Y$.

\subsection{TBS: Triangular Block SYRK}
\label{subsec:tbs}

The proof of Theorem~\ref{thm:syrk.main.bound} shows that the largest
operational intensity in the SYRK kernel is achieved when computing
the elements of $\mat{C}$ in a triangle $T(m)$, which is located at
the top-left of matrix $\mat{C}$. The result of
Corollary~\ref{thm:syrk.final.bound} is tight if all (or at least
most) parts of the computation have the same operational
intensity. But it is not clear whether it is possible to tile the
whole computation space with triangles. It is easy around the
diagonal, but what about the elements of the matrix away from the
diagonal?

\begin{algorithm}
  Partition $C$ in \emph{blocks} of size $S$\;
  \For{each block $B$}{
    Load the corresponding elements of $C$ in memory\;
    \For{$i=0$ to $M-1$}{
      Load the required elements of $A[\cdot,i]$\;
      Update block $B$ with these elements\;
    }
    Remove block $B$ from memory\;
  }
  \caption{Generic out-of-core SYRK algorithm}
  \label{alg:generic.syrk}
\end{algorithm}

Our algorithm uses the generic scheme described in
Algorithm~\ref{alg:generic.syrk}: store a block of elements of the
result matrix in memory, and iteratively load elements from $A$ to
update this block. To maximize memory efficiency, it makes sense that blocks 
would contain $S$ elements. In the \OOCSYRK algorithm proposed by
B\'ereux, the blocks are squares of $\sqrt{S}\times \sqrt{S}$, which is
the optimal shape without data reuse (for example, squares are the
optimal shape for non-symmetric GEMM multiplication). As mentioned
above, in order to match the lower bound for SYRK, we need to have
blocks shaped as triangles, up to row and column reordering:
such blocks are \emph{triangle blocks} $\TB(R)$, as defined in
Definition~\ref{def:triangle.block}. Indeed, $\TB(R)$ is the set of
indices of the elements of $C$ that can be updated with elements of
$A$ whose row belong in $R$.

We prove here that it is actually possible to tile (almost all) the
result matrix \mat{C} with triangle blocks, each containing roughly
$S$ elements.

\subsubsection{Partitioning the result matrix}

We fix $k$ such that
$$S \geqslant k+ \frac{k(k-1)}{2} = \frac{k(k+1)}{2}.$$ This ensures
that the memory can fit a triangle of side length $k$ from the result
matrix \mat{C}, plus a vector of $k$ elements of \mat{A} used for the
update. Let us assume for the moment that $N = ck$ for some value
$c$. We will see later that not all values of $c$ are eligible, and we
will discuss how to choose an appropriate value. We decompose the
result matrix $\mat{C}$ in $\frac{k(k-1)}{2}$ square \emph{zones} of
size $c \times c$. The rest of the matrix ($k$ triangle-shaped zones
on the diagonal) will be considered later. In \TBS, a block contains
exactly one element from each of these square zones, as depicted in
Figure~\ref{fig:tbs.zones}. For $0\leqslant i, j < c$, we denote by
$B^{i,j}$ the block which contains the element $(i, j)$ of the
top-most zone (which is the element $(i+c, j)$ of the matrix
$\mat{C}$).

Let $R^{i,j}$ be the row indices of block $B^{i,j}$. Since we search
for blocks with one element per zone, we can write
\begin{equation}
  B^{i,j} = \TB(R^{i,j}), \;\text{ with }\; R^{i,j} = \set{u\cdot c +
    f^{i,j}(u)}{0 \leqslant u < k}, \label{eq:def.blocks}
\end{equation}
where $0 \leqslant f^{i,j}(u) < c$ gives the position of the row of
$B^{i,j}$ within the $u$-th row of zones (see
left of Figure~\ref{fig:tbs.final}). To ensure that $B^{i,j}$
contains $(i+c, j)$, we just need to have $f^{i,j}(0) = j$ and
$f^{i,j}(1) = i$. We can thus specify our triangle blocks with an
\emph{indexing family}:

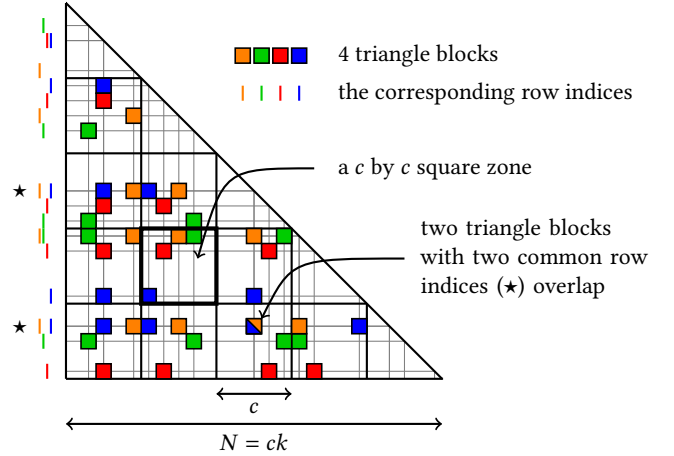
\begin{figure}
  \newcommand{\block}[3][0]{
    \begin{scope}
      \foreach \u [count=\i from 0] in {#2} {
        \draw[#3] (-0.2-#1*0.05, \i + 1-\u/5) -- (-0.2-#1*0.05, \i+1-\u/5-0.2);
        \begin{scope}[on background layer]
          \draw[gray,thin] (0, \i + 1-\u/5 - 0.1) -- (5-\i-1+\u/5+0.1,
          \i + 1-\u/5 - 0.1) -- (5-\i-1+\u/5+0.1, 0);
        \end{scope}

      }
      \foreach \u [count=\i from 0] in {#2} {
        \foreach \v [count=\j from 0] in {#2} {
          \ifthenelse{\i < \j}{
            \draw[semithick,fill=#3] (4-\j+\v/5, \i +1- \u/5) rectangle
            ( 4-\j + \v/5+0.2, \i+1-\u/5-0.2);
          }
        }
      }
    \end{scope}
  }
  \begin{tikzpicture}[thick]
    % Draw triangle with zones
    \foreach \i in {0, ..., 5} {
      \draw (\i, 0) -- (\i, 5 - \i);
      \draw (0, \i) -- (5 - \i, \i);
    }
    \draw (5, 0) -- (0, 5);

    % Legend
    \draw[<->] (2, -0.2) -- (3, -0.2) node[midway,below]{$c$};
    \draw[<->] (0, -0.6) -- (5, -0.6) node[midway,below]{$N=ck$};

    % block
    \block[0]{1, 4, 2, 0, 2}{blue}
    \block[1]{4, 1, 3, 1, 2}{red}
    \block[2]{2, 0, 4, 3, 1}{green!80!black}
    \block[3]{1, 0, 2, 2, 4}{orange}

    % legend for blocks
    \begin{scope}[yshift=3.8cm]
      \foreach \x/\color in {2.25/orange, 2.5/green!80!black, 2.75/red, 3/blue} {
        \draw[fill=\color,semithick] (\x,0.4) rectangle (\x+0.2, 0.6);
        \draw[color=\color] (\x+0.1,-0.1) -- (\x+0.1, 0.1);
      }
      \node[anchor=west] at (3.5, 0.5) {4 triangle blocks};
      \node[anchor=west] at (3.5, 0) {the corresponding row indices};
    \end{scope}
    % legend for zones
    \draw[ultra thick] (1, 1) rectangle (2, 2);
    \node[anchor=west] at (3.5, 2.8) {a $c$ by $c$ square zone};
    \draw[<-] (1.75, 1.6) ..controls (2.1, 2.8) .. (3.3, 2.8);

    % legend for conflict
    \draw[semithick, fill=blue, line join=bevel] (4-2+2/5, 0 +1- 1/5) |-
    ( 4-2 + 2/5+0.2, 0+1-1/5-0.2) --cycle;
    \draw[<-] (4-2+2/5+0.2, 0+1-1/5) ..controls (3, 1.6) .. (4.5, 1.6);
      \node[anchor=west, text width=3.2cm] at (4.6, 1.6) {two triangle
        blocks with two common row indices ($\star$) overlap};
      \node[anchor=east] at (-0.4, 0+1-1/5-0.1) {$\star$};
      \node[anchor=east] at (-0.4, 2 +1- 2/5-0.1) {$\star$};
  \end{tikzpicture}
  \caption{Zones and blocks in the \TBS algorithm. Each block has one
    element in each zone.}
  \label{fig:tbs.zones}
\end{figure}

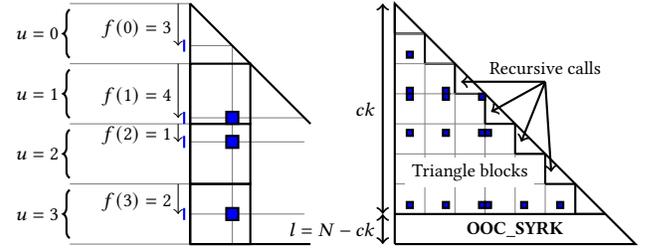
\begin{figure}
  \newcommand{\block}[3][0.4]{
    \begin{scope}
      \foreach \u [count=\i from 0] in {#2} {
        \foreach \v [count=\j from 0] in {#2} {
          \ifthenelse{\i < \j}{
            \draw[fill=#3] (3.5-#1-#1*\j+#1*\v/5, #1*\i +#1- #1*\u/5) rectangle
            ( 3.5-#1-#1*\j + #1*\v/5+#1*0.2, #1*\i+#1-#1*\u/5-#1*0.2);
          }
        }
      }
    \end{scope}
  }
  \begin{tikzpicture}[thick,node font=\footnotesize,scale=0.8]
    \foreach[count=\i from 0] \u in {4, ..., 1} {
      \draw[thin,gray] (-2, \u) -- (0, \u);
      \draw[snake=brace] (-2, \u-1+0.1) -- (-2, \u-0.1)
      node[midway,left=.5mm] {$u=\i$};
      \ifthenelse{\i > 0} {
        \draw (0, \u) rectangle (1, \u-1);
      }
    }
    \draw[thin,gray] (-2, 0) -- (0, 0);
    \draw[line join=bevel] (2, 2) -- (0, 4) -- (0, 0) -- (2, 0);

    \begin{scope}
      \pgfmathsetmacro{\ll}{3/5+0.1}
      \foreach \u [count=\i from 0] in {2, 1, 4, 3} {
        \pgfmathsetmacro{\index}{int(3-\i)}
        \draw[blue] (-0.1, \i + 1-\u/5) -- (-0.1, \i+1-\u/5-0.2);
        \draw[thin,->] (-0.2, \i+1) -- (-0.2, \i + 1 - \u/5-0.1)
        node[pos=0.6,left,node font=\footnotesize]{$f(\index) = \u$};
        \begin{scope}[on background layer]
          \draw[gray,thin] (0, \i + 1-\u/5 - 0.1) -- (\ll, \i + 1-\u/5 - 0.1);%% -- (5-\i-1+\u/5+0.1, 0);
        \end{scope}
      }
      \begin{scope}[on background layer]
        \draw[gray,thin] (\ll, 4-\ll) -- (\ll, 0);%% -- (5-\i-1+\u/5+0.1, 0);
      \end{scope}

      \pgfmathsetmacro{\v}{3}
      \foreach \u [count=\i from 0] in {2,1,4} {
        \draw[fill=blue] (0+\v/5, \i +1- \u/5) rectangle
        ( 0 + \v/5+0.2, \i+1-\u/5-0.2);
        \begin{scope}[on background layer]
          \draw[gray,thin] (\ll, \i + 1-\u/5 - 0.1) -- (1.9, \i + 1-\u/5 - 0.1);%% -- (5-\i-1+\u/5+0.1, 0);
        \end{scope}
      }
    \end{scope}

    \begin{scope}[xshift=3.4cm,yshift=0.5cm]
    % Draw triangle with zones
    \foreach \i in {0, 0.5, ..., 3} {
      \draw[gray, very thin] (\i, 0) -- (\i, 3.5 - \i);
      \draw[gray, very thin] (0, \i) -- (3.5 - \i, \i);
      \draw[thick] (\i, 3-\i) -- (\i+0.5, 3-\i);
      \draw[thick] (\i, 3.5-\i) -- (\i, 3-\i);
    }
    % Outer triangle
    \draw[thick] (0, -0.5) -- (4, -0.5) -- (0, 3.5) -- (0, -0.5);
    \draw[thick] (0, 0) -- (3.5, 0);

    \block[0.5]{3, 2, 1, 0, 4, 3, 2}{blue}

    % Legend
    \draw[<->] (-0.2, 0) -- (-0.2, 3.5) node[midway,left]{$ck$};
    \draw[<->] (-0.2, 0) -- (-0.2, -0.5) node[midway,left]{$l=N-ck$};

    \node at (2, -0.25) {\OOCSYRK};
    \node[fill=white] at (1.25, 0.7) {Triangle blocks};
    \node[anchor=south] at (2.5,  2.2) {Recursive calls};

    \foreach \i in {1, 1.5, ..., 2.5} {
      \draw[<-] (\i+0.1, 3.2-\i) -- (2.5, 2.2);
    }
    \end{scope}
  \end{tikzpicture}
  \caption{{\em Left:} $f^{i,j}(u)$ gives the position of the row of
    $B^{i,j}$ within the $u$-th row of zones. {\em Right:} which parts of the matrix $\mat{C}$ are computed
    by which method in the \TBS Algorithm.}
  \label{fig:tbs.final}
\end{figure}

\begin{definition}[Indexing family]
  A $(c, k)$-\emph{indexing family} is a family of functions $f^{i, j}(u)$, defined for
  $(i,j)$ in $\range{c-1}^2$, with:
  \begin{align*}
    & f^{i,j}: \range{k-1} \mapsto \range{c-1}\\
    \forall i, j, \quad & f^{i, j}(0) = j \quad\text{and}\quad f^{i, j}(1) = i\\
  \end{align*}
\end{definition}

To enforce the validity of the algorithm, triangle blocks $B^{i,j}$
must not overlap. If two functions $f^{i,j}$ and $f^{i',
  j'}$ agree for two different values $u$ and $v$, the corresponding
blocks $B^{i,j}$ and $B^{i',j'}$ have two row indices in common, and as can be seen on
Figure~\ref{fig:tbs.zones}, these blocks are not disjoint. We thus need to
consider \emph{valid} indexing families:

\begin{definition}[Validity]
  A $(c, k)$-indexing family $f$ is \emph{valid} if
  $$\forall u \neq v, \begin{cases}f^{i, j}(u) &= f^{i', j'}(u)\\ f^{i,
    j}(v) &= f^{i', j'}(v)\\
  \end{cases}\quad
  \Longrightarrow i = i' \text{ and } j
    = j'.$$
\end{definition}

It turns out that this condition is sufficient to ensure
no collisions:
\begin{lemma}
  If $f$ is a valid $(c, k)$-indexing family, then the sets $B^{i,j}$
  defined in Equation~\ref{eq:def.blocks} are pairwise
  disjoint.
\end{lemma}
\begin{proof}
We prove the contrapositive of this statement: if two $B^{i,j}$ sets
are not disjoint, then $f$ is not valid. Indeed, let us consider two
different pairs $(i,j)$ and $(i', j')$ such that $B^{i,j}\cap
B^{i',j'} \neq \emptyset$. There exist $(u, v)$ and $(u', v')$, with
$u\neq v$ and $u'\neq v'$, such that:
\begin{align*}
  uc + f^{i, j}(u) &= u'c + f^{i', j'}(u') \\
  vc + f^{i, j}(v) &= v'c + f^{i', j'}(v') 
\end{align*}
Since the values of an indexing function are in \range{c-1}, this
implies $u=u'$ and $v=v'$.

Thus, there exist $u\neq v$ and $i, j, i', j'$, with $(i, j) \neq (i',
j')$, such that $f^{i, j}(u) = f^{i', j}(u)$ and $f^{i, j}(v) = f^{i',
j'}(v)$: $f$ is not valid.
\end{proof}

This shows that using a valid indexing family allows to partition
the square zones from Figure~\ref{fig:tbs.zones} in disjoint triangle
blocks. The remaining elements, from the triangular zones close to the
diagonal, can be computed by recursive calls to the \TBS algorithm.
We thus require several valid indexing families for a fixed
$k$ and different values of $c$, since the recursive calls will be made with
different value of $c$. However, we see below that we cannot obtain
valid indexing families for all values of $c$, so we are not yet
ready to describe the complete algorithm.

\subsubsection{Defining a valid indexing family}

In this section, we show that it is possible to define a valid
indexing family for some values of $c \geqslant k-1$. We do this
using simple modulo operations:

\begin{definition}
  The \emph{cyclic} $(c,k)$-indexing family is defined by:
  $$f^{i,j}_c(u) =
  \begin{cases}
    j & \text{if $u = 0$}\\
    i + j(u-1) \mod c& \text{if $u$ > 0}
  \end{cases}
  $$
  \label{def:cyclic.function}
\end{definition}

\begin{lemma}
  If $c \geqslant k-1$ is coprime with all integers in \range[2]{k-2},
  then the cyclic indexing family $f_c$ is valid.
  \label{lemma:cyclic.is.valid}
\end{lemma}
\begin{proof}
  Consider any $u, v \in \range{k-1}$ with $u < v$, and assume that
  $i, j, i', j'$ in \range{c-1} are such that $f^{i,j}_c(u) =
  f^{i',j'}_c(u)$ and $f^{i,j}_c(v) = f^{i',j'}_c(v)$.

  We first prove $j = j'$. If $u=0$, this is direct. Otherwise, we can
  write

  \begin{align*}
    &\begin{cases}
      i + j(u-1) &= i'+j'(u-1) \mod c\\ i + j(v-1) &= i'+j'(v-1) \mod
      c
    \end{cases}\\
    \Leftrightarrow &\begin{cases}
      i-i' &= (j'-j)(u-1) \mod c\\
      i-i' &= (j'-j)(v-1) \mod c
    \end{cases}
  \end{align*}
  This implies:
  \begin{align*}
    (j'-j)(u-1) &= (j'-j)(v-1) &&\mod c\\
    \Leftrightarrow (j'-j)(u-v) &= 0 &&\mod c
  \end{align*}

  Since $u < v$, $0 < u,v \leqslant k-1$, we know that $0 < v-u
  \leqslant k-2$. From our assumption, $v-u$ is coprime with $c$, so
  we obtain $j'-j=0 \mod c$, and thus $j=j'$.

  Then, since $i + j(v-1) = i' + j(v-1) \mod c$, we deduce $i = i'\mod
  c$. Since $i,i'$ are in \range{c-1}, we have $i=i'$.
\end{proof}

We define the constant integer $q$ as the product of all primes no
larger than $k-2$: $q =\displaystyle \prod_{p \text{ prime}, p\leqslant
  k-2}p$. Then $c$ is coprime with all integers in \range[2]{k-2} if
and only if $c$ is coprime with $q$. Notice that $q$ is constant: it
only depends on $k$, thus on $S$, but not on $N$ or $M$.

\begin{algorithm}
  \KwIn{matrices $\mat{A}$ of size $N\times M$ and $\mat{C}$ symmetric of size $N\times N$}
  \KwOut{$\mat{C} \pluseq \mat{A}\cdot\transpose{\mat{A}}$}
  \KwAssume{memory of size $S = \frac{k(k+1)}{2}$}
  $q \gets \text{product of all primes in \range[2]{k-2}}$\;
  $c\gets \text{the largest integer coprime with $q$ below }\frac{N}{k}$\;
  $l\gets N - ck$\;
  %\KwRequire{a family of valid $(c, k)$-indexing functions $f_c$}
  \uIf(\comment*[f]{$c$ is too small}){$c < k-1$}{
    \OOCSYRK$(\mat{A}, \mat{C})$;
  }
  \Else{
  Use \OOCSYRK to compute the last $l$ rows of $\mat{C}$\;
  \For(\comment*[f]{recursive calls for triangular zones}){$i=0$ to
    $k-1$}{
    $R \gets \range[ic]{(i+1)c}$\;
    $\text{\TBS}(\mat{A}[R, \cdot], \mat{C}[R, R])$
    }
  \For(\comment*[f]{loop over all blocks}){$(i, j) \in \range{c-1}^2$}{
    $R \gets \set{r_u = uc + f^{i,j}_c(u)}{0 \leqslant u <
      k}$\comment*[r]{see Def.~\ref{def:cyclic.function}}
    Load the elements of C indexed by $\TB(R)$\;
    \For(\comment*[f]{loop over columns of $\mat{A}$}){$i=0$ to $M-1$}{
      Load elements of $\mat{A}$ indexed by $\set{(r,i)}{r \in R}$\;
      \For(\comment*[f]{loops over elements}){$u=0$ to $k-1$}{
        \For(\comment*[f]{ of the block}){$v=0$ to $u-1$}{
          $\mat{C}_{r_u, r_v} \pluseq \mat{A}_{r_u, i}\cdot \mat{A}_{r_v, i}$\;
        }
      }
    }
  }
  }
  \caption{Triangle Block Syrk $\TBS(\mat{A}, \mat{C})$}
  \label{alg:tbs.final}
\end{algorithm}

Now that we know how to build valid indexing families, we are ready to
describe the \TBS algorithm. However, with the constraints on $c$
imposed by Lemma~\ref{lemma:cyclic.is.valid}, it is not possible to
use triangle blocks on the whole matrix $\mat{C}$. Instead, given a
matrix size $N$, we set $c$ to be the largest number coprime with $q$
such that $c\leqslant \frac{N}{k}$. If the obtained $c$ is lower than
$k-1$, we can use the simple \OOCSYRK with square blocks. Otherwise,
$c$ satisfies the condition of Lemma~\ref{lemma:cyclic.is.valid}, so
we can use triangle blocks to compute the first $ck$ rows of $C$, and
the \OOCSYRK algorithm for the remaining $l = N-ck$ rows (see
right of Figure~\ref{fig:tbs.final}). The resulting algorithm is called \TBS,
and is described in Algorithm~\ref{alg:tbs.final}.

\subsubsection{Communication cost analysis}

Let us first notice that the \TBS algorithm loads each entry of $C$ exactly
once (even for the elements computed with \OOCSYRK), so loading these
elements has a communication cost of $\frac{N^2}{2}$. In the
following, we denote by $\qtyvals{\tbs}{N, M}$ the communication cost
of \TBS related to elements of $\mat{A}$, for a matrix $\mat{A}$ of size $N\times
M$.

The definition of $c$ yields $\frac{N}{k} = c + g$, and we need an
upper bound on $g$ to estimate the amount of work performed by
\OOCSYRK. It is easy to see that for any integer $a$, $aq + 1$ is
coprime with $q$. In particular, $\floor{N/kq}q + 1$ is coprime with
$q$, thus $c \geqslant \floor{N/kq}q + 1$, and $g \leqslant q$. Since
$q$ only depends on $S$ and not on $N$ or $M$, we get $g =
\bigo{1}$. Even though $q$ is a constant, it may be considered very
large relative to $S$. However, the bound $g\leqslant q$ is very
pessimistic: sieve methods allow to show that the number of integers
coprime with $q$ in any interval \range[(a-1)q]{aq-1} is exactly
$\prod (p-1)$, where $p$ spans the prime numbers below $k-1$ (see
Example~1.5 in~\cite{opera.de.cribro}). In practice, one can expect
the value of $g$ to be much lower than $q$.

We first consider the elements computed with the \TBS algorithm (in the
first $ck$ rows). There are $c^2$ triangle blocks, and each triangle block loads $kM$
elements of $\mat{A}$. This yields a communication cost $Q_1 = c^2kM$, and
with $c\leqslant \frac{N}{k}$, we obtain $Q_1 \leqslant \frac{N^2M}{k}$.

Elements computed with \OOCSYRK (in the last $l=gk$ rows) are
computed by square $\sqrt{S}\times\sqrt{S}$ blocks, and each block
loads $2M\sqrt{S}$ elements from matrix $\mat{A}$. Since there are at most
$gkN$ such elements, this yields a communication cost $Q_2 \leqslant
\frac{gkN}{S}\cdot 2M\sqrt{S}=\bigo{NM}$.

Adding the elements covered by the recursive calls, we get:
$$\qtyvals{\tbs}{N, M} \leqslant \frac{N^2M}{k} +
k\qtyvals{\tbs}{\frac{N}{k}, M} + \bigo{NM}$$

We can iteratively apply this inequality $t$ times, where $t$ is the
smallest integer such that $\frac{N}{k^t}< k-1$. We thus have
$k^{t-1}<\frac{N}{k}$, and $t = \bigo{\log N}$. Then we get:
\begin{align*}
  \qtyvals{\tbs}{N, M} &\leqslant
  \sum_{i=1}^{t}\frac{N^2M}{k^i} +
  k^t\qtyvals{\oocsyrk}{k, M} + t\cdot \bigo{NM}\\
  &\leqslant \sum_{i=1}^{\infty}\frac{N^2M}{k^i} +
  N\cdot\frac{k^2M}{\sqrt{S}} + \bigo{NM\log N}\\
   &\leqslant N^2M(\frac{1}{1-\frac{1}{k}} - 1) +
  \bigo{NM\log N}\\
   &\leqslant \frac{N^2M}{k-1} +\bigo{NM\log N}\\
\end{align*}

Remember that $k$ is defined by $S = \frac{k(k+1)}{2}$, so
that $k-1 \simeq \sqrt{2S}$.
%% \todo[inline]{This is imprecise, because we never assume $S$ tends to
%%   $\infty$. Should we correct this?}
%%   %O Je suis partagé...ça ne sert à rien sur le fond mais c'est quand même un peu faux... et je ne vois pas comment le corriger
In total (with the communications required to load elements of $\mat{C}$), we get:

\begin{theorem}
  The total communication cost $\qty{\tbs}{N, M}$ of the \TBS
  algorithm for a matrix $\mat{A}$ of size $N\times M$, with a memory of
  size $S$, is bounded by:

 $$\qty{\tbs}{N, M} \leqslant
  \frac{1}{\sqrt{2}}\cdot\frac{N^2M}{\sqrt{S}} + \frac{N^2}{2} +
  \bigo{NM\log N}$$
\end{theorem}

\subsubsection{Tiled version of \TBS}

The \TBS algorithm, as presented in Algorithm~\ref{alg:tbs.final}
achieves an asymptotic complexity which matches the lower bound from
Theorem~\ref{thm:syrk.final.bound}. However, this requires very large
values of $N$, since the condition $c \geqslant k-1$ together with $k
\simeq \sqrt{2S}$ means that the triangular block approach can only be
used for $N \geqslant 2S$. In that case, the matrix is so large that
half a column does not fit in memory.

To make the \TBS algorithm more practical, it is possible to design a
\emph{tiled} version of it, where elements of $\mat{C}$ are no longer
considered individually, but as tiles of size $b\times b$. We thus choose
$b$ and $k$ such that $S = b^2\frac{k(k-1)}{2}$, and set $c =
\frac{N}{kb}$ (actually the largest integer coprime with $q$ below
this value). Instead of loading elements of $C$, we thus load complete
tiles; however we still load elements of $\mat{A}$ one row at a
time. The update operation $\mat{C}_{r_u, r_v} \pluseq \mat{A}_{r_u,
  i}\cdot \mat{A}_{r_v, i}$ thus becomes an outer product.

The communication cost analysis is very similar, only the value of
$Q_1$ changes. There are still $c^2$ blocks, each of which loads $kbM$
elements of \mat{A}. We get $Q_1 = c^2kbM$, with $c \leqslant
\frac{N}{kb}$. Thus $Q_1 \leqslant \frac{N^2M}{kb}$. In turn, this
yields $\qtyvals{\tbs}{N, M} \leqslant \frac{N^2M}{(k-1)b} +
\bigo{NM\log N}$. With $b = \sqrt{\frac{2S}{k(k-1)}}$, we get:
$$\qty{\tbs}{N, M} \leqslant \frac{N^2M}{\sqrt{2S\cdot\frac{k-1}{k}}} +
\frac{N^2}{2} + \bigo{NM\log N}.$$
  
The leading term is now larger than the lower bound by a factor
$\sqrt{{k}/{(k-1)}}$, but this tiled version of the algorithm is
valid for smaller values of $N$. Indeed, the constraint $c\geqslant k-1$
implies $N \geqslant \frac{2S}{b} = \sqrt{2S\cdot k(k-1)}$, and thus
$\frac{N^2}{2} \geqslant k(k-1)$: \TBS is useful as soon as storing the matrix
requires $k(k-1)$ times more memory than available.

\subsection{\LBC: Large Block Cholesky}
\label{subsec:lbc}

The lower bound detailed in Section~\ref{subsec:chol.lower.bound}
is based on the idea that Cholesky factorization generates at least as
many data transfers as SYRK operation. Since \tbs algorithm
performs the SYRK kernel with the minimum amount of I/O operations,
the idea is to use it for the largest possible part of the computation
of the Cholesky factorization.

\subsubsection{Algorithm description}
We implement this strategy in the \textbf{Large Block Cholesky} (\lbc)
algorithm. It is a right-looking, blocked algorithm which performs the
Cholesky factorization of any input symmetric positive definite matrix
$\mat{A}$ making use of \OOCCHOL, \OOCTRSM and \TBS algorithms.
Note that it would be possible to use
  a recursive call to \LBC instead of \OOCCHOL, since \LBC performs
  fewer transfers. However, it turns out that the successive Cholesky
  factorizations of $\mat{A}[I_0, I_0]$ do not contribute to the
  higher order term, so we opt for \OOCCHOL to simplify the
  presentation. \lbc modifies $\mat{A}$
in-place to yield a lower triangular matrix $\mat{L}$ as output such
that $\mat{A} = \mat{L}\cdot\transpose{\mat{L}}$. The steps of the
algorithm are detailed in Algorithm~\ref{algo:lbc} and described on Figure~\ref{fig:chol_2x2}.

\begin{figure}
  \begin{center}
    \begin{tikzpicture}[thick, scale=0.4]

      % Zones <=> algorithm
      \filldraw[fill=black!35] (0, 0) -- (0, 10) -- (2, 8) -- (2, 0) -- cycle ;
      \filldraw[fill=red!50] (2, 6) -- (2, 8) -- (4, 6) -- cycle ;
      \filldraw[fill=green!50] (2, 0) -- (2, 6) -- (4, 6) -- (4, 0) -- cycle ;
      \filldraw[fill=blue!50] (4, 0) -- (4, 6) -- (10, 0) -- cycle ;

      % Triangle
      \draw[very thick] (0, 0) -- (0, 10) -- (10, 0) -- cycle;
      % Vertical lines
      \draw (2, 8) -- (2, 0);
      \draw (4, 6) -- (4, 0);
      % Horizontal lines
      \draw (2, 6) -- (4, 6);

      % i.b position
      \draw[dashed] (2, 0) -- (2, -0.8);
      \node [below] () at (2, -0.8) {$i \cdot b$};
      \draw[dashed] (-0.2, 8) -- (2, 8);
      \node [left] () at (-0.2, 8) {$i \cdot b$};

      % Vertical arrows
      \draw[<->] (1.8, 8) -- (1.8, 6) node[midway, left]{$b$};
      % Horizontal arrows
      \draw[<->] (2, 5.8) -- (4, 5.8) node[midway, below]{$b$};
      \draw[<->] (2, -0.2) -- (4, -0.2) node[midway, below]{$I_{0}$};
      \draw[<->] (4, -0.2) -- (10, -0.2) node[midway, below]{$I_{1}$};
      % \draw[<->] (0, -1.5) -- (10, -1.5) node[midway, below]{$N$};

      % A zones naming
      \node [anchor=west] () at (2.8, 7.3) {$\mat{A}[I_{0}, I_{0}]$};
      \node [anchor=east] () at (-1, 3) {$\mat{A}[I_{1}, I_{0}]$};
      \draw [thin, ->] (-0.9, 3) -- (3, 3); 
      \node [] () at (6, 2) {$\mat{A}[I_{1}, I_{1}]$};

      % Algorithm legend
      \coordinate (L) at (8, 10);
      \filldraw[draw, fill=red!50] (L) rectangle +(0.8, -0.8);
      \node [anchor=west] () at ([shift={(1, -0.4)}] L) {\OOCCHOL};
      \filldraw[draw, fill=green!50] (L)+(0, -1) rectangle ([shift={(0.8, -1.8)}] L);
      \node [anchor=west] () at ([shift={(1, -1.4)}] L) {\OOCTRSM};
      \filldraw[draw, fill=blue!50] (L)+(0, -2) rectangle ([shift={(0.8, -2.8)}] L);
      \node [anchor=west] () at ([shift={(1, -2.4)}] L) {\TBS};
    
    \end{tikzpicture}
  \end{center}
  \caption{Algorithm LBC: updating the three parts of $\mat{A}$ at iteration $i$}
  \label{fig:chol_2x2}
\end{figure}
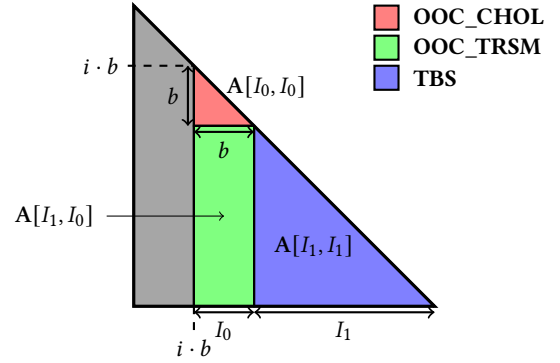

\begin{algorithm}
  \KwIn{$\mat{A}$: $N \times N$ symmetric positive definite matrix}
  \KwIn{$b$: block size}
  \KwAssume{$b | N$}

  \KwOut{$\mat{L}$: $N \times N$ lower triangular matrix s.t. $\mat{A} = {\mat{L}}\cdot\transpose{\mat{L}}$}
  %%  \BlankLine
  \For(){$i=0$ to $\lfloor \frac{N}{b} \rfloor$}{
    $I_{0} = \range[i \cdot b]{(i+1) \cdot b}$\;
    $\mat{A}[I_{0}, I_{0}] \leftarrow$ \OOCCHOL$(\mat{A}[I_{0}, I_{0}])$\;
    \If{$(i+1) \cdot b < N$}{
      $I_{1} = \range[(i+1) \cdot b]{N}$\;
      $\mat{A}[I_{1}, I_{0}] \leftarrow$ \OOCTRSM$(\mat{A}[I_{0}, I_{0}], \mat{A}[I_{1}, I_{0}])$\;
      $\mat{A}[I_{1}, I_{1}] \leftarrow$ \TBS$(\mat{A}[I_{1}, I_{0}], \mat{A}[I_{1}, I_{1}])$\;
    }
  }
  \caption{Large Block Cholesky $\LBC(\mat{A})$}
  \label{algo:lbc}
\end{algorithm}

\LBC is a so-called \emph{right-looking} variant of Cholesky
factorization. At each iteration, the final values of the two leftmost
panels $\mat{A}[I_{0}, I_{0}]$ and $\mat{A}[I_{1}, I_{0}]$ are
computed; $\mat{A}[I_{1}, I_{0}]$ is then used to update the
right panel $\mat{A}[I_{1}, I_{1}]$ whose values are still
temporary. By contrast, \emph{left-looking variants} perform all the
updates of a given value of $\mat{A}$ one after the other, allowing to
write each element only once.

Right-looking implementations of Cholesky are known to perform more
I/O operations than their left-looking counterparts, because the lower
right panel $\mat{A}_{I_{1}, I_{1}}$ needs to be reloaded at each
iteration, so as to be updated using the SYRK kernel. Nevertheless,
this overhead can be rendered negligible. Indeed, the main point of
\LBC is to use large enough blocks (of size $\sqrt{N}$), so that the
number of iterations is low ($\sqrt{N}$): then, the volume of communications induced
by loading $\mat{A}_{I_{1}, I_{1}}$ remain negligible compared to
the one required to update its values.

%Such variant is know to require more data transfers than its \emph{left-looking} counterpart because of the SYRK operation performed at each iteration to update the right-lower panel. Indeed, at each step $i$ there are $(\frac{N-(i \cdot b)}{2})^{2}$ elements in $\mathbf{A}[I_{1}, I_{1}]$ to update, each requiring access to $2b$ values in $\mathbf{A}[I_{1}, I_{0}]$ (except for diagonal elements). Summed over all iterations, it implies $\mathcal{O}((\frac{N^{5}}{b^{4}}))$ data transfers.
%
%On the other hand, the calculation of all final values in $\mathbf{A}[I_{1}, I_{0}]$ using TRSM operation requires at each step $\mathcal{O}(b^{2}N)$ data transfers. This yields a total of $\mathcal{O}(N^{2}b)$ communications for the whole factorization.
%
%Hence block size $b$ must be selected carefully for a blocked right-looking algorithm, such as \textbf{LBS}, to reach even asymptotically the lower bound on number of communications.

\subsubsection{Communication cost analysis}
Let us now analyze the total number of I/O operations required by
Algorithm~\LBC on an $N \times N$ matrix $\mat{A}$; it is
denoted $\qty{\lbc}{N}$. As mentioned above, we get
from~\cite{doi:10.1137/06067256X} that $\qty{\ooctrsm}{N, M} =
  \frac{N^2M}{\sqrt{S}} + \bigo{NM}$ and $\qty{\oocchol}{N} =
  \frac{N^3}{3\sqrt{S}} + \bigo{NM}$.
%% \begin{align*}
%%   \qty{\ooctrsm}{N, M} =
%%   \frac{N^2M}{\sqrt{S}} + \bigo{NM} \quad && \qty{\oocchol}{N} =
%%   \frac{N^3}{3\sqrt{S}} + \bigo{NM}.
%% \end{align*}
Furthermore, as detailed in
\ref{subsec:tbs}, we also know that $\qty{\tbs}{N, M} =
\frac{1}{\sqrt{2}} \frac{N^{2}M}{\sqrt{S}} +\frac{N^{2}}{2}
+\bigo{NM\log N}$. Then:

\begin{align*}
  \qty{\lbc}{N}
  & = \displaystyle\sum_{i = 1}^{\frac{N}{b}} {\qty{\oocchol}{b}}
  +{\qty{\ooctrsm}{b, \left(\tfrac{N}{b}-i\right)b}}
  +{\qty{\tbs}{\left(\tfrac{N}{b}-i\right)b, b}}\\
  & = \frac{N}{b} \qty{\oocchol}{b}
  +\displaystyle\sum_{i = 1}^{\frac{N}{b}}
  \qty{\ooctrsm}{b, ib}
  +\qty{\tbs}{ib, b} \\
  \begin{split}
    &= \frac{b^{2}N}{3\sqrt{S}}+\bigo{b^{2}} +\\
    &\qquad +\displaystyle\sum_{i = 1}^{\frac{N}{b}} \left(
    \frac{b^{2} (ib)}{\sqrt{S}}
    +\frac{b (ib)^{2}}{\sqrt{2}\sqrt{S}}+\frac{(ib)^{2}}{2}+\bigo{b^{2}
      i\log(ib)} \right)
  \end{split}
\end{align*}
%  & = \frac{b^{2}N}{3\sqrt{S}} + \frac{b^{3}(\frac{N}{b})^{2}}{2\sqrt{S}} + \frac{b^{3}(\frac{N}{b})^{3}}{3\sqrt{2}\sqrt{S}} + \frac{b^{2}(\frac{N}{b})^{3}}{6}

Since $0 < b < N$, $\bigo{b^{2}} = \bigo{N^{2}} $.

Besides:
$$\sum_{i = 1}^{\frac{N}{b}} \bigo{b^{2} i\log(ib)}
  \leqslant \sum_{i = 1}^{\frac{N}{b}} \bigo{b^{2} \tfrac{N}{b}\log N}
  = \frac{N}{b} \bigo{Nb\log N}
  = \bigo{N^{2}\log N}$$.

The number of data transfers necessary to perform algorithm \LBC is therefore:

\begin{align*}
  \qty{\lbc}{N}
  & \leqslant \frac{b^{2}N}{3\sqrt{S}}
  +\sum_{i = 1}^{\frac{N}{b}} \left(
  \frac{b^{2} (ib)}{\sqrt{S}}
  +\frac{b (ib)^{2}}{\sqrt{2}\sqrt{S}}+\frac{(ib)^{2}}{2} \right)
  +\bigo{N^{2}\log N} \\
  & \leqslant \frac{b^{2}N}{3\sqrt{S}}
  +\frac{b^{3}(\frac{N}{b})^{2}}{2\sqrt{S}} %+o((\frac{N}{b})^{2})
  +\frac{b^{3}(\frac{N}{b})^{3}}{3\sqrt{2}\sqrt{S}} %+o((\frac{N}{b})^{3})
  +\frac{b^{2}(\frac{N}{b})^{3}}{6} %+o((\frac{N}{b})^{2})
  +\bigo{N^{2}\log N} \\
  & \leqslant \underbrace{\frac{b^{2}N}{3\sqrt{S}}}_{(1)}
  +\underbrace{\frac{b N^{2}}{2\sqrt{S}}}_{(2)}
  +\underbrace{\frac{N^{3}}{3\sqrt{2}\sqrt{S}}}_{(3)}
  +\underbrace{\frac{\frac{N^{3}}{b}}{6}}_{(4)}
  +\; \bigo{N^{2}\log N}
\end{align*}

As previously discussed the volume of data transfers induced by
loading $\mat{A}_{I_{1}, I_{1}}$ at each step (4) clearly becomes
dominant if $b$ is a constant. On the other hand, if the chosen value
for $b$ is of order $N$, the communications required to perform all
TRSM operations (2) becomes dominant. Hence, to ensure that the volume
of data transfers used for $\mat{A}_{I_{1}, I_{1}}$ update (3) is the
only dominant term in the formula, we choose to implement \LBC using $b
= \sqrt{N}$ as block size. Then:

\begin{align*}
  \qty{\lbc}{N}
  &\leqslant \frac{N^2}{3\sqrt{S}}
  +\frac{N^{2}\sqrt{N}}{2\sqrt{S}}
  +\frac{N^{3}}{3\sqrt{2}\sqrt{S}}
  +\frac{N^{2}\sqrt{N}}{6}
  +\bigo{N^{2}\log N}\\
  %% & = \frac{N^{3}}{3\sqrt{2}\sqrt{S}}
  %% +(\frac{1}{2\sqrt{S}}+\frac{1}{6}) N^{2}\sqrt{N}
  %% +\frac{N^{2}}{3\sqrt{S}}
  %% +\bigo{N^2} \\
  & = \frac{N^{3}}{3\sqrt{2}\sqrt{S}}
  +\bigo{N^{5/2}}
\end{align*}

\begin{theorem}
  The total communication cost $\qty{\lbc}{N}$ of the \LBC
  algorithm for a matrix $\mat{A}$ of size $N\times N$, with a memory of
  size $S$, is bounded by:

 $$\qty{\lbc}{N} \leqslant
  \frac{1}{3\sqrt{2}}\cdot\frac{N^3}{\sqrt{S}} + \bigo{N^{\frac{5}{2}}}.$$
\end{theorem}

%% Asymptotically, the communication cost of the \LBC algorithm therefore
%% meets the lower bound stated in
%% Corollary~\ref{thm:cholesky.final.bound}.

\section{Conclusion}

This paper provides a definitive answer to the asymptotic
communication complexities of both the SYRK and Cholesky kernels. The
perhaps surprising answer is that the symmetric nature of these
computations can actually be taken advantage of, so that their
operational intensities are intrinsically higher than those of their
non-symmetric counterparts (matrix multiplication and LU
factorization). In addition to our theoretical lower bound results,
our algorithms provide insights about how to make use of the symmetry
to reduce communications. In future works, it might be possible to
improve the lower order terms of our results, to obtain efficient
algorithms for not so large values of $N$. More importantly, we
believe that the insight provided by this paper can be a starting
point to obtain communication efficient parallel algorithms for
symmetric linear algebra kernels. Finally, our work might also be
extended to other kernels which use the same input several times.

%TODO Ne pas oublier de le remettre pour la version finale.
\begin{acks}
This work is supported in part by the Région Nouvelle-Aquitaine, under grant
  2018-1R50119 "HPC scalable ecosystem", and by the ANR, under grant SOLHARIS - ANR-19-CE46-0009.
\end{acks}

\bibliographystyle{ACM-Reference-Format}
\bibliography{./biblio,./article}


\begin{thebibliography}{12}

%%% the Web address.

\ifx \showCODEN    \undefined \def \showCODEN     #1{\unskip}     \fi
\ifx \showDOI      \undefined \def \showDOI       #1{#1}\fi
\ifx \showISBNx    \undefined \def \showISBNx     #1{\unskip}     \fi
\ifx \showISBNxiii \undefined \def \showISBNxiii  #1{\unskip}     \fi
\ifx \showISSN     \undefined \def \showISSN      #1{\unskip}     \fi
\ifx \showLCCN     \undefined \def \showLCCN      #1{\unskip}     \fi
\ifx \shownote     \undefined \def \shownote      #1{#1}          \fi
\ifx \showarticletitle \undefined \def \showarticletitle #1{#1}   \fi
\ifx \showURL      \undefined \def \showURL       {\relax}        \fi
% invisible to TeX
\providecommand\bibfield[2]{#2}
\providecommand\bibinfo[2]{#2}
\providecommand\natexlab[1]{#1}
\providecommand\showeprint[2][]{arXiv:#2}

\bibitem[\protect\citeauthoryear{Ballard, Carson, Demmel, Hoemmen, Knight, and
  Schwartz}{Ballard et~al\mbox{.}}{2014}]%
        {ballard_carson_demmel_hoemmen_knight_schwartz_2014}
\bibfield{author}{\bibinfo{person}{G. Ballard}, \bibinfo{person}{E. Carson},
  \bibinfo{person}{J. Demmel}, \bibinfo{person}{M. Hoemmen},
  \bibinfo{person}{N. Knight}, {and} \bibinfo{person}{O. Schwartz}.}
  \bibinfo{year}{2014}\natexlab{}.
\newblock \showarticletitle{Communication lower bounds and optimal algorithms
  for numerical linear algebra}.
\newblock \bibinfo{journal}{\emph{Acta Numerica}}  \bibinfo{volume}{23}
  (\bibinfo{year}{2014}), \bibinfo{pages}{1–155}.
\newblock
\urldef\tempurl%
\url{https://doi.org/10.1017/S0962492914000038}
\showDOI{\tempurl}


\bibitem[\protect\citeauthoryear{Ballard, Demmel, Holtz, and Schwartz}{Ballard
  et~al\mbox{.}}{2010}]%
        {ballard2010communication}
\bibfield{author}{\bibinfo{person}{Grey Ballard}, \bibinfo{person}{James
  Demmel}, \bibinfo{person}{Olga Holtz}, {and} \bibinfo{person}{Oded
  Schwartz}.} \bibinfo{year}{2010}\natexlab{}.
\newblock \showarticletitle{Communication-optimal parallel and sequential
  {Cholesky} decomposition}.
\newblock \bibinfo{journal}{\emph{SIAM Journal on Scientific Computing}}
  \bibinfo{volume}{32}, \bibinfo{number}{6} (\bibinfo{year}{2010}),
  \bibinfo{pages}{3495--3523}.
\newblock


\bibitem[\protect\citeauthoryear{Ballard, Demmel, Holtz, and Schwartz}{Ballard
  et~al\mbox{.}}{2011}]%
        {DBLP:journals/siammax/BallardDHS11}
\bibfield{author}{\bibinfo{person}{Grey Ballard}, \bibinfo{person}{James
  Demmel}, \bibinfo{person}{Olga Holtz}, {and} \bibinfo{person}{Oded
  Schwartz}.} \bibinfo{year}{2011}\natexlab{}.
\newblock \showarticletitle{Minimizing Communication in Numerical Linear
  Algebra}.
\newblock \bibinfo{journal}{\emph{SIAM J. Matrix Analysis Applications}}
  \bibinfo{volume}{32}, \bibinfo{number}{3} (\bibinfo{year}{2011}),
  \bibinfo{pages}{866--901}.
\newblock
\urldef\tempurl%
\url{https://doi.org/10.1137/090769156}
\showDOI{\tempurl}


\bibitem[\protect\citeauthoryear{B{\'e}reux}{B{\'e}reux}{2009}]%
        {doi:10.1137/06067256X}
\bibfield{author}{\bibinfo{person}{Natacha B{\'e}reux}.}
  \bibinfo{year}{2009}\natexlab{}.
\newblock \showarticletitle{Out-of-Core Implementations of {Cholesky}
  Factorization: Loop-Based versus Recursive Algorithms}.
\newblock \bibinfo{journal}{\emph{SIAM J. Matrix Anal. Appl.}}
  \bibinfo{volume}{30}, \bibinfo{number}{4} (\bibinfo{year}{2009}),
  \bibinfo{pages}{1302--1319}.
\newblock
\urldef\tempurl%
\url{https://doi.org/10.1137/06067256X}
\showDOI{\tempurl}
\showeprint{https://doi.org/10.1137/06067256X}


\bibitem[\protect\citeauthoryear{Friedlander and Iwaniec}{Friedlander and
  Iwaniec}{2010}]%
        {opera.de.cribro}
\bibfield{author}{\bibinfo{person}{John Friedlander} {and}
  \bibinfo{person}{Henryk Iwaniec}.} \bibinfo{year}{2010}\natexlab{}.
\newblock \bibinfo{booktitle}{\emph{Opera de cribro}}.
  \bibinfo{series}{American Mathematical Society Colloquium Publications},
  Vol.~\bibinfo{volume}{57}.
\newblock \bibinfo{publisher}{American Mathematical Society},
  \bibinfo{address}{Providence, RI}.
\newblock
\showISBNx{978-0-8218-4970-5}
\urldef\tempurl%
\url{https://doi.org/10.1090/coll/057}
\showDOI{\tempurl}


\bibitem[\protect\citeauthoryear{Irony, Toledo, and Tiskin}{Irony
  et~al\mbox{.}}{2004}]%
        {irony2004communication}
\bibfield{author}{\bibinfo{person}{Dror Irony}, \bibinfo{person}{Sivan Toledo},
  {and} \bibinfo{person}{Alexander Tiskin}.} \bibinfo{year}{2004}\natexlab{}.
\newblock \showarticletitle{Communication lower bounds for distributed-memory
  matrix multiplication}.
\newblock \bibinfo{journal}{\emph{J. Parallel and Distrib. Comput.}}
  \bibinfo{volume}{64}, \bibinfo{number}{9} (\bibinfo{year}{2004}),
  \bibinfo{pages}{1017--1026}.
\newblock


\bibitem[\protect\citeauthoryear{Jia-Wei and Kung}{Jia-Wei and Kung}{1981}]%
        {Jia-Wei:1981:ICR:800076.802486}
\bibfield{author}{\bibinfo{person}{Hong Jia-Wei} {and} \bibinfo{person}{H.~T.
  Kung}.} \bibinfo{year}{1981}\natexlab{}.
\newblock \showarticletitle{{I/O} Complexity: The Red-blue Pebble Game}. In
  \bibinfo{booktitle}{\emph{Proceedings of the Thirteenth Annual ACM Symposium
  on Theory of Computing}} (Milwaukee, Wisconsin, USA)
  \emph{(\bibinfo{series}{STOC '81})}. \bibinfo{publisher}{ACM},
  \bibinfo{address}{New York, NY, USA}, \bibinfo{pages}{326--333}.
\newblock
\urldef\tempurl%
\url{https://doi.org/10.1145/800076.802486}
\showDOI{\tempurl}


\bibitem[\protect\citeauthoryear{Kwasniewski, Kabic, Ben-Nun, Ziogas, Saethre,
  Gaillard, Schneider, Besta, Kozhevnikov, VandeVondele, and
  Hoefler}{Kwasniewski et~al\mbox{.}}{2021}]%
        {kwasniewski2021parallel2}
\bibfield{author}{\bibinfo{person}{Grzegorz Kwasniewski},
  \bibinfo{person}{Marko Kabic}, \bibinfo{person}{Tal Ben-Nun},
  \bibinfo{person}{Alexandros~Nikolaos Ziogas}, \bibinfo{person}{Jens~Eirik
  Saethre}, \bibinfo{person}{Andr\'{e} Gaillard}, \bibinfo{person}{Timo
  Schneider}, \bibinfo{person}{Maciej Besta}, \bibinfo{person}{Anton
  Kozhevnikov}, \bibinfo{person}{Joost VandeVondele}, {and}
  \bibinfo{person}{Torsten Hoefler}.} \bibinfo{year}{2021}\natexlab{}.
\newblock \showarticletitle{On the Parallel I/O Optimality of Linear Algebra
  Kernels: Near-Optimal Matrix Factorizations}. In
  \bibinfo{booktitle}{\emph{Proceedings of the International Conference for
  High Performance Computing, Networking, Storage and Analysis}} (St. Louis,
  Missouri) \emph{(\bibinfo{series}{SC '21})}. \bibinfo{publisher}{Association
  for Computing Machinery}, \bibinfo{address}{New York, NY, USA}, Article
  \bibinfo{articleno}{70}, \bibinfo{numpages}{15}~pages.
\newblock
\showISBNx{9781450384421}
\urldef\tempurl%
\url{https://doi.org/10.1145/3458817.3476167}
\showDOI{\tempurl}


\bibitem[\protect\citeauthoryear{Olivry, Iooss, Tollenaere, Rountev,
  Sadayappan, and Rastello}{Olivry et~al\mbox{.}}{2021}]%
        {olivry:hal-03200539}
\bibfield{author}{\bibinfo{person}{Auguste Olivry}, \bibinfo{person}{Guillaume
  Iooss}, \bibinfo{person}{Nicolas Tollenaere}, \bibinfo{person}{Atanas
  Rountev}, \bibinfo{person}{P Sadayappan}, {and} \bibinfo{person}{Fabrice
  Rastello}.} \bibinfo{year}{2021}\natexlab{}.
\newblock \showarticletitle{{IOOpt: Automatic Derivation of I/O Complexity
  Bounds for Affine Programs}}. In \bibinfo{booktitle}{\emph{{PLDI 2021 - 42nd
  ACM SIGPLAN International Conference on Programming Language Design and
  Implementation}}}. \bibinfo{address}{Virtual, Canada}.
\newblock
\urldef\tempurl%
\url{https://doi.org/10.1145/3453483}
\showDOI{\tempurl}


\bibitem[\protect\citeauthoryear{Olivry, Langou, Pouchet, Sadayappan, and
  Rastello}{Olivry et~al\mbox{.}}{2020}]%
        {olivry2020automated}
\bibfield{author}{\bibinfo{person}{Auguste Olivry}, \bibinfo{person}{Julien
  Langou}, \bibinfo{person}{Louis-No{\"e}l Pouchet}, \bibinfo{person}{P
  Sadayappan}, {and} \bibinfo{person}{Fabrice Rastello}.}
  \bibinfo{year}{2020}\natexlab{}.
\newblock \showarticletitle{Automated derivation of parametric data movement
  lower bounds for affine programs}. In \bibinfo{booktitle}{\emph{Proceedings
  of the 41st ACM SIGPLAN Conference on Programming Language Design and
  Implementation}}. \bibinfo{pages}{808--822}.
\newblock


\bibitem[\protect\citeauthoryear{Solomonik, Carson, Knight, and
  Demmel}{Solomonik et~al\mbox{.}}{2017}]%
        {10.1145/2897188}
\bibfield{author}{\bibinfo{person}{Edgar Solomonik}, \bibinfo{person}{Erin
  Carson}, \bibinfo{person}{Nicholas Knight}, {and} \bibinfo{person}{James
  Demmel}.} \bibinfo{year}{2017}\natexlab{}.
\newblock \showarticletitle{Trade-Offs Between Synchronization, Communication,
  and Computation in Parallel Linear Algebra Computations}.
\newblock \bibinfo{journal}{\emph{ACM Trans. Parallel Comput.}}
  \bibinfo{volume}{3}, \bibinfo{number}{1}, Article \bibinfo{articleno}{3}
  (\bibinfo{date}{jan} \bibinfo{year}{2017}), \bibinfo{numpages}{47}~pages.
\newblock
\showISSN{2329-4949}
\urldef\tempurl%
\url{https://doi.org/10.1145/2897188}
\showDOI{\tempurl}


\bibitem[\protect\citeauthoryear{Solomonik and Demmel}{Solomonik and
  Demmel}{2011}]%
        {solomonik2011lu25D}
\bibfield{author}{\bibinfo{person}{Edgar Solomonik} {and}
  \bibinfo{person}{James Demmel}.} \bibinfo{year}{2011}\natexlab{}.
\newblock \showarticletitle{Communication-Optimal Parallel {2.5D} Matrix
  Multiplication and {LU} Factorization Algorithms}. \bibinfo{pages}{90--109}.
\newblock
\showISBNx{978-3-642-23396-8}
\urldef\tempurl%
\url{https://doi.org/10.1007/978-3-642-23397-5_10}
\showDOI{\tempurl}


\end{thebibliography}


%% \section{Additional Proofs}
\end{document}